\documentclass[11pt,a4paper]{article}
\usepackage{amssymb,amsthm}
\usepackage{fullpage}
\usepackage{amsmath}
\usepackage{mathrsfs}
\usepackage{txfonts}
\usepackage{lmodern}
\usepackage{graphicx}
\usepackage{hyperref}
\usepackage{comment}
\usepackage{enumerate}
\usepackage{extarrows}
\usepackage{tikz-cd}
\usepackage[arrow,matrix,curve]{xy}

\def\CT{{CT}}
\newcommand{\nil}{\mathrm{Nil}}
\newcommand{\Ad}{\mathrm{Ad}}
\newcommand{\heisZ}{\mathrm{Heis}^\ZZ}
\newcommand{\heisR}{\mathrm{Heis}^\RR}

\newcommand{\real}{{\mathbb R}}
\newcommand{\complex}{{\mathbb C}}
\newcommand{\torus}{{\mathbb T}}
\newcommand{\integer}{{\mathbb Z}}

\newcommand{\cA}{{\mathcal A}}

\newcommand{\cb}{{\mathcal B}}

\newcommand{\ce}{{\mathcal E}}

\newcommand{\ci}{{\mathcal I}}
\newcommand{\cL}{{\mathcal L}}

\newcommand{\ad}{{\rm ad}}
\newcommand{\cpt}{{\mathcal K}}
\newcommand{\cK}{{\mathcal K}}

\newcommand{\cT}{{\mathcal T}}
\newcommand{\cX}{{\mathcal X}}
\newcommand{\ind}{{\rm Ind}}

\newcommand{\inn}{{\rm Inn}}
\newcommand{\aut}{{\rm Aut}}

\newcommand{\wt}[1]{\widetilde{#1}}
\newcommand{\wh}[1]{\widehat{#1}}

\newcommand\rmap[1]{\stackrel{#1}\longrightarrow}
\newcommand\lmap[1]{\stackrel{#1}\longleftarrow}
\newcommand\lumap[1]{\vcenter{\llap{\scriptsize$#1$}}\uparrow}
\newcommand\rdmap[1]{\downarrow\vcenter{\rlap{\scriptsize$#1$}}}
\newcommand\ldmap[1]{\vcenter{\llap{\scriptsize$#1$}}\downarrow}

\newcommand{\CC}{{\mathbb C}}
\newcommand{\RR}{{\mathbb R}}
\newcommand{\TT}{{\mathbb T}}
\newcommand{\ZZ}{{\mathbb Z}}

\newcommand{\Ind}{\mathrm{Ind}}

\newcommand{\solR}{\mathrm{Sol}^\RR}
\newcommand{\solZ}{\mathrm{Sol}^\ZZ}
\newcommand{\solv}{\mathrm{Solv}}
\newcommand{\bbK}{{\mathbb K}}
\newcommand{\bbQ}{{\mathbb Q}}
\newcommand{\bbR}{{\mathbb R}}

\theoremstyle{plain}
\newtheorem{theorem}{Theorem}[section]
\newtheorem{lemma}[theorem]{Lemma}
\newtheorem{proposition}[theorem]{Proposition}
\newtheorem{corollary}[theorem]{Corollary}

\theoremstyle{definition}

\newtheorem{definition*}{Definition}

\theoremstyle{remark}
\newtheorem{remark}[theorem]{Remark}

\newtheorem{remarks*}{Remarks}
\newtheorem{axioms}[theorem]{Axiom}
\newtheorem*{setup}{Basic setup}

\begin{document}

\title{T-duality simplifies bulk-boundary correspondence:\\
The noncommutative case}

\author{Keith C. Hannabuss\footnote{{Mathematical Institute, 
  Andrew Wiles Building,
  Radcliffe Observatory Quarter,
  Woodstock Road,
  Oxford OX2 6GG,
U.K.,} {\em email}: kch@balliol.ox.ac.uk}, Varghese Mathai\footnote{ Department of Pure Mathematics,
School of  Mathematical Sciences, 
University of Adelaide, 
Adelaide, SA 5005, 
Australia, {\em email}: mathai.varghese@adelaide.edu.au}
 and Guo Chuan Thiang\footnote{ Department of Pure Mathematics,
School of  Mathematical Sciences, 
University of Adelaide, 
Adelaide, SA 5005, 
Australia, {\em email}: guochuan.thiang@adelaide.edu.au}
}

\date{}

\maketitle

\begin{abstract}
We state and prove a general result establishing that T-duality, or the Connes-Thom isomorphism simplifies the bulk-boundary correspondence, given by a boundary map in $K$-theory, in the sense of converting it to a simple geometric restriction map. This settles in the affirmative several earlier conjectures of the authors, and provides a clear geometric picture of the correspondence. In particular, our result holds in arbitrary spatial dimension, in both the real and complex cases, and also in the presence of disorder, magnetic fields, and H-flux. These special cases are relevant both to String Theory and to the study of the quantum Hall effect and topological insulators with defects in Condensed Matter Physics.
\end{abstract}

%
%

\tableofcontents

\section{Introduction: physical motivation}

A succession of insights, starting with Kane--Mele's invariant \cite{KM} and culminating in \cite{Ki,FM,Thiang}, led to an understanding and classification of topological insulators in terms of $K$-theory. The connection to $K$-theory is particularly important in view of earlier rigorous work on the quantum Hall effect \cite{Bellissard,Kellendonk1}. These developments also provided new insight into how the $K$-theory of bulk and boundary systems are related. For various reasons it is actually closer to the spirit of condensed matter theory, as well as more flexible and appropriate, to use the more general $C^*$-algebraic $K$-theory rather than the topological version \cite {B}. The availability of $C^*$-algebraic machinery also opens the way to describe the same system in different forms, using powerful tools such as T-duality and the closely related Connes--Thom isomorphism. As we explain in this paper, especially in the Appendix, the (topological) T-duality transform used in physical applications is precisely the Connes--Thom isomorphism \cite{C} composed with an isomorphism from an imprimitivity theorem.

Whilst the geometric relationship between bulk and boundary is most simply expressed in position space, the physics of Brillouin zones and Fermi levels is more transparent in momentum space. T-duality is a geometric analogue of the Fourier transform, see Section \ref{sect:overview}. It moves between geometric and physical descriptions similar to the Fourier transform that flips position and momentum space in wave mechanics. A recent series of papers has shown that the rather complicated map relating the 
$K$-theory of the physical bulk and boundary algebras of topological insulators, transforms under T-duality into the obvious geometrical restriction map \cite{MT1,MT2,MT3,HMT}. Schematically, these results say that the following diagram commutes in many examples of physical interest:
\begin{equation}
\xymatrix{
\framebox{\parbox{13em}{Position space bulk invariant}}  \ar[dd]_{\parbox{5em}{\footnotesize Restriction to boundary}} \ar[rrr]^{\sim\;}_{\rm T-duality} &&& \framebox{\parbox{14em}{Momentum space bulk invariant}} \ar[dd]^{\parbox{6em}{\footnotesize bulk-boundary homomorphism}} \\ &&& \\
\framebox{\parbox{15em}{Position space boundary invariant}}  \ar[rrr]^{\sim\;\;\;\;}_{\rm T-duality\;\;\;\;\;} &&& \framebox{\parbox{16.5em}{Momentum space boundary invariant}}
}
\label{metadiagram}
\end{equation}
Thus momentum space topological invariants (e.g.\ as defined over a Brilluoin torus) transform isomorphically under T-duality into real (position) space\footnote{Quite often $KR$-theory groups, in the sense of Atiyah's Real $K$-theory, are relevant, and these groups are defined for spaces with involutions. Such spaces are usually called ``Real spaces'', which should not be confused with our usage of ``real space'' as synonymous with ``position space''.} invariants, in such a way that the bulk-boundary homomorphism becomes the simple geometric restriction-to-boundary map.

The main result of this paper is the formalization and \emph{conceptual} proof of the commutativity of diagram \ref{metadiagram} in a very general setting. This result can then be specialized to analyze the bulk-boundary correspondence in: (1) arbitrary spatial dimensions, (2) in both the real and complex cases, (3) in the presence of disorder and/or magnetic fields, and (4) in the presence of H-flux. We had previously worked in some special subcases in which the groups of topological invariants in all four corners of diagram \ref{metadiagram}, or at least the interesting ones, are computable; then commutativity was verified by chasing around the diagram generator-by-generator. Getting a concrete hold of the generators can be very useful when studying particular systems. Nevertheless there are situations, particularly when disorder is modelled in the mathematical description, in which the relevant topological invariants are not easily computable. Our abstract proof then becomes valuable and represents a significant advance compared to our earlier results.

The physical importance of the first three cases is well-known and was also discussed in our earlier papers, so let us briefly elaborate on the fourth case. H-fluxes live in degree-3 cohomology and are important in string theory, although it seldom appears in condensed matter physics. However, as noted in \cite{HMT}, they can be used to describe screw dislocations. More precisely, defects such as screw dislocations break the $\ZZ^{2d+1}$ translation symmetry of the original $(2d+1)$D Euclidean lattice underlying a $(2d+1)$D insulator, so that standard constructions like unit cells and Brillouin zones are not available. This means that topological invariants in the presence of defects are not easily defined using standard methods appealing to Bloch theory. We can circumvent this difficulty by requiring the defects to be distributed in a regular manner, allowing us to define topological invariants which are, in a precise sense, deformed versions of the usual ones arising from Bloch theory \cite{HMT}. Indeed, the non-isotropy introduced by the defects is expected to be reflected in boundary phenomena in the form of gapless modes localized along the dislocations \cite{RZV}.
\begin{figure}[h]
    \centering
    \includegraphics[width=0.5\textwidth]{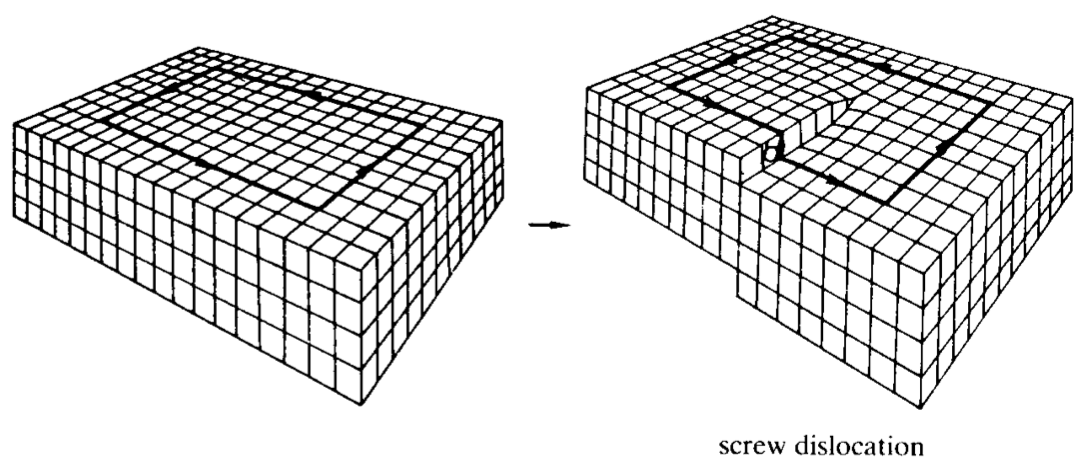}
    \caption{\em \small An elementary screw dislocation (Source: pp.\ 786 in \cite{Kleinert}) oriented in the vertical direction. A circuit of horizontal translations enclosing the dislocation ends at a lattice site which differs from the starting point by a vertical translation.}
    \label{fig:screw}
\end{figure}

There may be other applications, e.g.\ Dirac monopoles are known to be described by 3-cocycles \cite{WZ}. Although monopoles have proved elusive, analogues have been created in condensed matter. Bramwell et al. have detected them through the magnetic Wien effect in spin ices, in particular, in single crystals of cubic Dy$_2$Ti$_2$O$_7$, \cite{Br}, whilst Ray et al. have produced detailed observations in a ${}^{87}$Rb Bose--Einstein condensate, \cite{Ray}.

Torsion H-fluxes have appeared in condensed matter in earlier works  in slightly different contexts, by Freed-Moore \cite{FM} in the presence of symmetry 
and by Gawedzki \cite{Gawedzki} who considers the holonomy of a flat gerbe. In our context, possibly {\em non torsion} H-fluxes 
arise when one takes the partial T-dual, in the fibre direction, of a Heisenberg nilmanifold (associated to screw dislocated lattices mentioned earlier) 
which is thought of as a circle bundle over a torus. The T-dual is in this case another torus, with possibly non-torsion H-flux that is explained in 
Section \ref{sect:heis}. We mention that non torsion H-fluxes also arise naturally in the recent study of 4D semimetals by the 2nd and 3rd authors of this paper, in \cite{MT17}.

\subsection*{Outline}
Section \ref{sect:overview} provides an overview of various T-duality transformations, applied to the bulk-boundary correspondence in a variety of settings. The $C^*$-algebras for the bulk and boundary from both the geometrical and physical perspective are defined very generally in Section \ref{sect:BB}. Theorem \ref{thm:CTbulkboundary} and Corollary \ref{cor:restrictionequalsPV} formalizes the commutative diagram \ref{metadiagram} in this general setting. The Dixmier--Douady classes, or H-fluxes, provide the main focus in Section \ref{sect:FD}, and we prove our \cite[Conjecture 2.1]{HMT} in Theorem \ref{thm:productcase}. 
Section \ref{sect:heis}, which contains examples and also generalisations of the discrete Heisenberg groups and Heisenberg nilmanifolds, that play an important role in \cite{HMT} and often arise as symmetry groups in physical systems. 
The example of the three dimensional solvable group is also worked out here. In Section \ref{sect:disorder}, we study the bulk-boundary correspondence for the higher dimensional quantum Hall effect, as well as topological insulators with disorder in both the real and complex cases, generalising the 2D results of \cite{MT2} and the higher-dimensional results without disorder in \cite{MT3}. The Appendix explains why Paschke's map \cite{Pas} is essentially the same as the Connes--Thom isomorphism \cite{C}, in both the complex and real $C^*$-algebra cases. The Paschke map implements a T-duality transformation in 1D, and we derive his explicit formula starting from abstract arguments.

%
%

\section{Overview of T-duality applied to bulk-boundary correspondence}\label{sect:overview}
Formalizing and proving commutativity of diagram \ref{metadiagram} is significant because it
\begin{itemize}\vspace{-0.6em}
    \itemsep-0.5em
    \item rigorously captures the geometric intuition behind the bulk-boundary correspondence;
    \item simplifies the complicated (momentum space) boundary map defined at the level of the physical $C^*$-algebras of observables;
    \item continues to hold even when the momentum space is \emph{noncommutative} (as in the IQHE), in the presence of \emph{disorder} \cite{MT2}, in the \emph{real} case (relevant for time-reversal symmetry), and in the \emph{parametrised} and \emph{twisted} setting of \cite{HM09,HM10} (relevant to string theory in the presence of H-flux \cite{BEM}). \vspace{-0.6em}
\end{itemize}
At the mathematical level, the meta-principle expressed by the diagram naturally generalizes a simple phenomenon already present at the level of ordinary Fourier transforms: integration in Fourier space, which can be understood topologically as a push-forward map, picks out the constant Fourier mode, so it is equivalent to the Fourier transform of a ``restriction-to-zero'' map.
A simple example of how the ordinary Fourier transform acts on an integration map goes as follows. Write $({\bf n},n_d)=n\in\ZZ^d$ and let $\iota$ be the inclusion of $\ZZ^{d-1}\hookrightarrow\ZZ^d$ taking ${\bf n}\mapsto({\bf n},0)$. 
The Fourier transform, \, $\widehat{}:f\mapsto\widehat{f}$ takes a rapidly decreasing function $f:\ZZ^d\rightarrow\CC$ to a smooth function $\widehat{f}:\widehat{\ZZ^d}=\TT^d\rightarrow\CC$, and is implemented by the Schwartz kernel $P({\bf n},{\bf k})=e^{2\pi{\rm i} {\bf n}\cdot {\bf k}},\, {\bf n}\in\ZZ^d, {\bf k}\in\TT^d$,
\begin{equation}
    \widehat{f}({\bf k})=\sum_{\bf n} P({\bf n},{\bf k})f({\bf n})=\sum_{\bf n} e^{2\pi {\rm i} {\bf n}\cdot {\bf k}}f({\bf n})\nonumber.
\end{equation}
Let $\partial:\widehat{f}\mapsto \partial\widehat{f}$ be integration along the $d$-th circle in $\TT^d$. This picks out only the part of $\widehat{f}$ with Fourier coefficient $n_d=0$, so there is a commutative diagram
\begin{equation}\label{FTdiagram}
\xymatrix{
f  \ar[d]^{\iota^*} \ar[rr]^{\sim\;}_{\widehat{}} && \widehat{f} \ar[d]^\partial \\
\iota^*f \ar[rr]^{\sim}_{\widehat{}} && \partial\widehat{f}}
\end{equation}
where $\iota^*$ is simply restriction to $n_d=0$. Our paper studies in particular, a very general noncommutative geometry analog of this phenomenon.

\subsection{Physics and geometry of the bulk-boundary correspondence}
The bulk-boundary correspondence refers, generally speaking, to the appearance of some boundary-localised phenomenon (e.g.\ quantised chiral edge currents) inherited from a corresponding bulk phenomenon (e.g.\ quantised Hall conductance). To account for such a correspondence, one begins a model of the bulk-with-boundary physical system. A natural and powerful method is to use a \emph{Toeplitz extension algebra} to ``glue'' together a boundary $C^*$-algebra of observables to a bulk $C^*$-algebra. Such a construction was used to analyse the bulk-boundary correspondence in the context of the Integer quantum Hall effect in \cite{Kellendonk1}, and generalised and explained in the monograph \cite{PSB} and Section 2 of \cite{MT3}. 

The idea behind the Toeplitz extension can be sketched quite easily. In a lattice model, a bulk Hamiltonian acts on a Hilbert space $l^2(\ZZ^d)\otimes V$, where $\ZZ^d$ labels lattice sites and $V\cong\CC^N$ captures the internal degrees of freedom. Writing $U^y, y\in\ZZ^d$ as the shift operators on $l^2(\ZZ^d)$, a lattice model Hamiltonian may be written as $H=\sum_{y\in\ZZ^d} U^y\otimes W_y$, where $W_y$ are $N\times N$ \emph{hopping matrices} satisfying $W_y^*=W_{-y}$. Let $U_i, i=1,\ldots,d$ denote the unit translation in the $i$-th direction, which together generate an algebra $C^*(\ZZ^d)=C^*(U_i,\ldots,U_d)$ isomorphic to $C(\TT^d)$ after Fourier transform. Note that $\TT^d$ here is a \emph{momentum space} torus. Then we see that $H$ is a self-adjoint element of some matrix algebra over $C^*(\ZZ^d)$, which we take as the \emph{bulk algebra} $\wh{\cb}$. If there is a boundary, transverse to the $d$-th direction say, then $U_d$ needs to be replaced by a unilateral shift $\widehat{U}_d$ satisfying $\widehat{U}_d^*\widehat{U}_d=1,\quad \widehat{U}_d\widehat{U}_d^*=1-e,$ where $e$ is projection onto $0$ in the $d$-th component. Note that $\widehat{U}_d$ is no longer unitary but is instead a Toeplitz operator, which together with $U_1,\ldots,U_{d-1}$ generates a \emph{Toeplitz} algebra $\cT$ as the bulk-with-boundary algebra, instead of $C^*(\ZZ^d)$ as in the bulk algebra. The two-sided ideal of $\cT$ generated by $e$ is a stabilised version of the algebra $C^*(U_1,\ldots,U_{d-1})$ generated by translations parallel to the boundary, and is thus identified as the \emph{boundary algebra} $\wh{\ce}$. In summary, the Toeplitz exact sequence $$
0 \longrightarrow \wh{\ce} \longrightarrow \cT \longrightarrow \wh{\cb} \longrightarrow 0,
$$ exhibits the bulk-with-boundary algebra $\cT$ as an extension of the bulk algebra by the boundary algebra. Associated to this sequence is a Toeplitz index map $K_\bullet(\wh{\cb})\xrightarrow{\partial} K_{\bullet-1}(\wh{\ce})$ which maps between the $K$-theories (topological invariants) of the bulk algebra and the boundary algebra.

If a Hamiltonian has $\ZZ^d$ translational symmetry (e.g.\ due to a periodic potential provided by an atomic lattice), one traditionally \emph{defines} bulk/boundary topological invariants in quasi-momentum space $\TT^d$. Namely, one constructs from a $\ZZ^d$-symmetric Hamiltonian, Bloch vector bundles of single-electon states over the Brillouin torus $\TT^d$, which in turn define elements in the $K$-theory of $\TT^d$. Such vector bundles correspond to projections in matrix algebras over the bulk algebra $C^*(\ZZ^d)\cong C(\TT^d)$ described earlier, and this extends in $K$-theory to $K^{-\bullet}(\TT^d)\cong K_\bullet(C^*(\ZZ^d))$; similarly for the boundary algebra. The bulk-to-boundary map $\partial$ is then essentially an integration (push-forward) map along the projection $\TT^d\rightarrow\TT^{d-1}$ \cite{MT2,MT3}. This suggests (as in \eqref{FTdiagram}) that a geometric Fourier transform (T-duality) can turn $\partial$ into a simple geometric restriction map. In more realistic models which incorporate magnetic fields and disorder, the Brillouin zone may be \emph{noncommutative} and the analogous construction yields elements in the $K$-theory of a twisted crossed product by $\ZZ^d$ \cite{B}. For instance, \emph{magnetic} translations commute only up to a phase determined by the magnetic flux, and they generate a noncommutative torus \cite{Bellissard, Kellendonk1} rather than the commutative algebra $C^*(\ZZ^d)$.

It is therefore customary to start with the physical (momentum space) picture involving $C^*$-algebras of observables, and then T-dualize to obtain a geometrical picture in real (position) space, but the exposition in this paper is a bit simpler if one starts with the geometric picture. In any case, the intuition behind the bulk-boundary correspondence in physics is largely a \emph{geometrical} one. For example, in a semiclassical picture of the IQHE, the electrons in the bulk are confined to cyclotron orbits in the plane of the sample due to the applied perpendicular magnetic field \cite{CN}. These circular orbits are intercepted by a geometrical boundary, leading to chiral edge currents, see Figure\ \ref{fig:cyclotron}.
\begin{figure}
    \centering
    \includegraphics[scale=0.3]{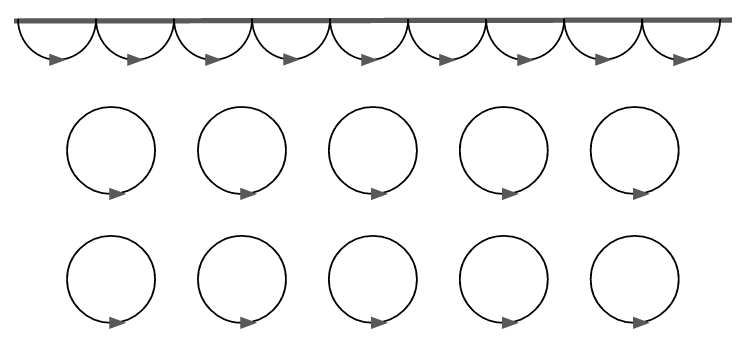}
    \caption{Cyclotron orbits intercepted by a boundary.}
    \label{fig:cyclotron}
\end{figure}
The chirality of the edge phenomena also highlights a crucial necessary ingredient in modelling bulk-boundary correspondences. Namely, although we are free to define topological invariants for the boundary system considered on its own, it is only from some intrinsic geometrical relationship between the boundary system and the bulk system that we can extract a notion of chirality. Furthermore, not all possible boundary invariants are relevant, but only those which ``lift'' to an invariant associated to the bulk system. Mathematically, this says that there should be a homomorphism mapping from bulk topological invariants to boundary ones, which need not be surjective or injective, and which is derived in a strong sense from some geometrical data about how the boundary is related to the bulk.

\subsubsection*{Exact sequences linking bulk and boundary algebras}
Geometrically we can visualise the bulk as a sort of finite or semi-infinite closed \lq\lq cylinder\rq\rq\  with cross-section the boundary, and the interior of the bulk as the corresponding open cylinder. We denote the $C^*$-algebras corresponding to the boundary, the bulk, and the interior as $\ce$, $\cb$, and $\ci$, respectively. (In the simplest cases they are just the functions on the geometric spaces defined by their spectra.) Then the interior, $\ci$, may be considered as the part of the bulk algebra \lq\lq vanishing\rq\rq\ on the boundary, and we have the exact sequence 
\begin{equation*}
0 \longrightarrow \ci \longrightarrow \cb \longrightarrow \ce \longrightarrow 0,
\end{equation*}
where $\cb \longrightarrow \ce$ is the geometrical restriction-to-boundary map. Standard theory enables us to extract from this short exact sequence an exact hexagon of the operator algebraic $K$-groups: 
$$
\begin{matrix}K_0(\ci) &\longrightarrow &K_0(\cb) 
&\longrightarrow  &K_0(\ce) \cr
 \uparrow & & & 
 &\downarrow
\cr
K_1(\ce) &\longleftarrow &K_1(\cb) &\longleftarrow &K_1(\ci). \cr\end{matrix}
$$

In the settings that we are interested in, the bulk has some $\ZZ$-symmetry transverse to the boundary, so that transversal shifts do not essentially change the geometry. For instance, one often studies Hamiltonians with the geometrical property of invariance under translations of a lattice (atomic positions) in the bulk of a material. Then the geometric relationship between the interior and boundary algebras, described in detail in the next section,  means that their $K$-theories differ only by a degree shift, 
$K_j(\ci) \cong K_{j+1}(\ce)$, so that 
\begin{equation}
\begin{matrix}K_1(\ce) &\longrightarrow &K_0(\cb) 
&\longrightarrow  &K_0(\ce) \cr
 \uparrow & & & 
 &\downarrow
\cr
K_1(\ce) &\longleftarrow &K_1(\cb) &\longleftarrow &K_0(\ce). \cr\end{matrix}\label{hexagon1}
\end{equation}

From the \emph{physical} perspective the bulk and boundary algebras of observables, denoted by $\wh{\cb}$, $\wh{\ce}$ respectively,  are the T-duals (Fourier transforms) of their geometric counterparts. (In the simplest cases, this is a is a ``momentum space'' perspective.) They are related by the fact that $\wh{\cb}$ and $\wh{\ce}$ fit in an extension,
$$
0 \longrightarrow \wh{\ce} \longrightarrow \cT \longrightarrow \wh{\cb} \longrightarrow 0,
$$
where $\cT$ is a Toeplitz-like algebra \cite{WO} associated with $\wh{\ce}$ and $\wh{\cb}$ \cite{PSB}. One can show that the $K$-theories of $\cT$ and $\wh{\ce}$ coincide, leading to another similar exact hexagon of K-groups:
\begin{equation}
\begin{matrix}K_0(\wh{\ce}) &\longrightarrow &K_0(\wh{\ce}) &\longrightarrow &K_0(\wh{\cb}) \cr
\lumap{\partial} & & & &\rdmap{\partial} \cr
K_1(\wh{\cb}) &\longleftarrow &K_1(\wh{\ce})&\longleftarrow &K_1(\wh{\ce}), \cr\end{matrix}\label{hexagon2}
\end{equation}
As explained in detail in \cite{PSB}, the Toeplitz extension algebra $\cT$ may be interpreted as the bulk-with-boundary algebra, and the vertical connecting maps $\partial$ in this diagram may be identified as the physical bulk-boundary homomorphisms.

As we explain in Section \ref{sect:BB}, T-duality simply interchanges the geometric and physical perspectives $\cb \leftrightarrow \wh{\cb}$, $\ce \leftrightarrow \wh{\ce}$. So far we have not said anything detailed about the maps in our exact hexagons, but we shall see that the geometric description is much simpler, and that the two hexagons encode the same information. In particular, the bulk-boundary homomorphism $\partial$ which we are interested in becomes the geometric restriction map under T-duality, as one might expect from heuristics.


The precise $K$-theory groups which are relevant for a given physical system vary in complexity, and accordingly, there are T-duality transformations of varying intricacy as summarized in Table \ref{Tdualitytable}. The first three rows had been dealt with in previous papers \cite{MT1,MT2,MT3}. For example, the 2D/3D Kane--Mele invariants \cite{KM} (for systems with fermionic time-reversal symmetry) live in (the reduced part of) $KQ^0(\TT^d,\varsigma)\cong KR^{-4}(\TT^d,\varsigma)$ with $d=2,3$, where $\varsigma$ is the momentum reversal involution $k\mapsto-k$, while their T-duals live in $KO^{-1}(\TT^3)$ or $KO^{-2}(\TT^2)$. The Kitaev Majorana chain \cite{Ki} invariants (with particle-hole symmetry) are in $KR^{-2}(\TT,\varsigma)$ with T-dual in $KO^{-1}(\TT)$ . In this paper, we focus on the last three rows, which subsume the first three.

\begin{table}
    \centering
\footnotesize
\begin{tabular}{p{8.5em}|p{9em}|p{10em}|p{13em}}
 \hline\hline
 Topological invariant & T-dual invariant & T-duality transformation & Physical examples \\
  \hline\hline
  $K^{-\bullet}(\TT)$ & $K^{-\bullet-1}(\TT)$ & Fourier--Mukai transform in 1D & Winding mode-momentum mode duality in string theory \\
  \hline
  $K^{-\bullet}(\TT^d)$ & $K^{-\bullet-d}(\TT^d)$ &  Fourier--Mukai transform in $d$-D & Chern insulators, Type II string theories\\
  \hline
  $KO^{-\bullet}(\TT^d)$ & $KR^{-\bullet-d}(\TT^d,\varsigma)$ & Real Fourier--Mukai transform & Time-reversal symmetric topological insulators, orientifold string theories\\
  \hline
    $K_{\bullet}(C(\TT^d))$ & $K_{\bullet+d}(A_\Theta)$ & Noncommutative T-duality & Integer quantum Hall effect\\
  \hline
    $K(O)_{\bullet}(\Ind_{\ZZ^d}^{\RR^d}(\mathcal{A},\alpha))$ & $K(O)_{\bullet+d}(\mathcal{A}\rtimes_\alpha\ZZ^d)$ & Paschke--Connes--Thom isomorphism & Disordered topological insulators\\
    \hline
  $K^{-\bullet}(X\times\TT^d, H)$;\;\;\;\;\; $H=H_1+H_2+H_3$ & $K_{\bullet+d}(\CT(Y_{H_2},H_3)_{H_1}$ & Connes--Thom isomorphism, Rieffel--Kasprzak deformation & String theory with H-flux, topological insulators with defects\\
  \hline
  \hline
    \end{tabular}
    \caption{Table of T-duality transformations}
\label{Tdualitytable}
\end{table}

\normalsize

%
%

\section{Bulk and boundary algebras}\label{sect:BB}
We proceed to construct the bulk and boundary algebras $\ce, \cb$, as well as their T-duals $\wh{\ce}, \wh{\cb}$, which appear in the long exact sequences \eqref{hexagon1}, \eqref{hexagon2} in the previous section. Let $L$, $M$, and $N= L\oplus M$ be lattices of maximal rank in the vector spaces $U\cong\RR$, $V\cong\RR^{d-1}$, and $W = U\oplus V\cong\RR^d$, with $d\geq 1$. $W$ will be associated with the bulk, $V$ with the boundary directions, and $U\cong\RR$ with the direction transverse to the boundary\footnote{Although we concentrate here on the case where the boundary has codimension 1, our notation and arguments have been written to suggest a way to include higher-codimensional boundaries. The main differences lie in the parity of the $K$-groups linked by Connes--Thom maps and the fact that there can be Mackey obstructions for the boundary.}. 

Usually, it is only the additive group structure of the spaces and lattices which appears in actions as algebra automorphisms. In all the physical examples that we will consider, $\ce$ is a $C^*$-algebra equipped with commuting actions $\tau: V\cong\RR^{d-1}\to \aut(\ce)$ and $\alpha: L\cong\ZZ\to \aut(\ce)$. Physically, $\tau$ is supposed to encode information about translational symmetries parallel to the boundary\footnote{There may be no such symmetries even if the boundary occupies $n\geq 1$ dimensions, for instance, it may not be completely straight. Generally speaking, the collection of boundary symmetries determine what topological invariants may be associated to it, and the more symmetries there are, the more boundary invariants are available.}, while $\alpha$ encodes a transverse $L$-covariance (e.g.\ with respect to the position of the boundary). We will often write $\alpha^\ell$ for the action of $\ell\in L \cong\integer$, so that $\alpha$ will be used both for the action and for its generator. 

As already noted, an obvious way to model the interior would be as an open cylinder on the boundary. We make this more precise by taking the interior algebra to be continuous $\ce$-valued functions on $(0,1)$, i.e.\ $\ci = C_0((0,1),\ce)$. Reflecting covariance with respect to $L$, the bulk algebra $\cb$ is defined to be the \emph{induced algebra} $\ind_L^{U}(\ce,\alpha)$, of continuous $\ce$-valued functions on $U$, which satisfy the equivariance condition $f(u - \ell) = \alpha^\ell [f(u)]$ for $u \in U$, $\ell \in L$; ($\ell$ counts the number of lattice spacings in $L$)\footnote{If $\alpha$ defined a character of $L$ then the \lq\lq periodicity\rq\rq\  would give Bloch-type functions.}. An equivalent definition of $\cb$ is as a \emph{mapping torus}:
\begin{equation}
T_\alpha(\ce) = \{ f\in C([0,1],\ce): f(1) = \alpha^{-1}[f(0)] \}.\label{mappingtorusequation}
\end{equation}
There is the obvious inclusion map $C_0((0,1),\ce) \overset{i}{\to} T_\alpha(\ce)\cong \ind_L^{U}(\ce,\alpha)$. There is automatically a translation action $\tau^\alpha$ of $U$ on $\cb = \ind_L^{U}(\ce,\alpha)$, which combines with the action of $V$ on $\ce$ to give an action of $W$ on $\cb$. The actions of $V$ and $W$ are used, respectively, to define the T-dual boundary and bulk algebras $\wh{\ce}$ and $\wh{\cb}$.

The basic setup is summarized as follows:
\begin{setup}\ 
\begin{itemize}
    \item $U=\RR$, $V=\RR^{d-1}$, $W=U\oplus V=\RR^d$ with standard lattices $L=\ZZ\subset U$, $M=\ZZ^{d-1}\subset V$, $N=\ZZ^d=L\oplus M\subset W$.
    \item The geometrical boundary algebra $\ce$ is a real or complex separable $C^*$-algebra.
    \item There are commuting continuous homomorphisms $\tau: V\to \aut(\ce)$ and $\alpha: L\to \aut(\ce)$ where $\aut(\ce)$ is given the topology of pointwise convergence.
    \item The geometrical bulk algebra $\cb$ is $\ind_L^{U}(\ce,\alpha)$, and so inherits from $\ce$ an action $\tau^\alpha\times\tau$ of $W=U\oplus V\cong \RR^d$.
    \item The T-dual boundary algebra is the crossed product $\wh{\ce} \coloneqq \ce \rtimes_\tau V$, while the T-dual bulk algebra is $\wh{\cb} \coloneqq \ce \rtimes_{\tau^\alpha\times\tau} W$.
\end{itemize}
\end{setup}
For notational ease, we will sometimes drop the action from the crossed product symbol when it is understood. This basic general setup will be specialized to physically interesting examples in later sections. Our arguments in this section work equally well for real or complex $\ce$ with appropriate modifications, so we will proceed with the assumption that all algebras are complex $C^*$-algebras, for which there are two complex $K$-theory groups $K_\bullet(\cdot),\,\bullet\in\ZZ_2$.

From the above setup, we have the ``geometric'' short exact sequence 
\begin{equation}
\begin{matrix}
0 &\rmap{} &\ci &\rmap{i} &\cb 
&\rmap{\varepsilon} &\ce &\rmap{}  &0\cr
 & &\| & &\| & & & &\cr
  & &C_0((0,1),\ce) & &\ind_{L}^{U}(\ce,\alpha). & & & &\cr\end{matrix}\label{geometricSES}
\end{equation}
where $\varepsilon: \ind_{L}^{U}(\ce,\alpha) \to \ce$ is just evaluation at 0, in other words, it restricts from the bulk to the boundary (at 0). Both $\ci,\cb$ inherit actions of $V$ and $L$ from $\ce$, and, by construction, the maps $i$ and $\varepsilon$ are equivariant for these actions.

The associated long exact sequence in $K$-theory is
$$
\begin{matrix}K_0(\ci) &\rmap{i_*} &K_0(\cb) 
&\rmap{\varepsilon_*} &K_0(\ce) \cr
 \lumap{} & & & 
 &\rdmap{}
\cr
K_1(\ce) &\lmap{\varepsilon_*} &K_1(\cb) &\lmap{i_*} &K_1(\ci). \cr\end{matrix}
$$
Since $\ci$ is isomorphic to the suspension $S\ce=C_0(\RR,\ce)$ of $\ce$, their $K$-groups agree apart from a degree shift. We also write $S$ for the $K$-theory isomorphisms $S:K_\bullet(S(\ce))\rightarrow K_{\bullet+1}(\ce),\,\bullet\in\ZZ_2$. With this identification, we can rewrite the diagram and the vertical connecting homomorphisms in the form
\begin{equation}\label{geometricLES}
\begin{matrix}K_1(\ce) &\rmap{i_*\circ S^{-1}} &K_0(\cb) 
&\rmap{\varepsilon_*} &K_0(\ce) \cr
 \lumap{1-\alpha^{-1}_*} & & & &\rdmap{1-\alpha^{-1}_*}\cr
K_1(\ce) &\lmap{\varepsilon_*} &K_1(\cb) &\lmap{i_*\circ S^{-1}}
 &K_0(\ce). \cr\end{matrix}
\end{equation}
where the subscript star indicates the induced maps on $K$-theory, and $(\alpha(1)^{-1})_*$ has been abbreviated to $\alpha^{-1}_*$ (c.f.\ Exercise 9.K of \cite{WO}).

We now take the crossed product of the geometric short exact sequence \eqref{geometricSES} by $V$ (an exact group), which acts on every term through its action on $\ce$. As in the basic setup, we write $\wh{\ce} = \ce \rtimes V$ for the T-dual boundary algebra.
The action $\alpha$ extends to an action $\wh{\alpha}$ on the crossed product $\wh{\ce}$, and $\cb \rtimes V \cong \ind_L^U(\wh{\ce},\wh{\alpha})$, since the action $\tau$ of $V$ on $\ce$ commutes with that of $L$. Similarly, $\ci\rtimes V \cong C_0((0,1),\wh{\ce}) \cong S(\wh{\ce})$, the suspension, so the effect of applying $\rtimes V$ to \eqref{geometricSES} is the short exact sequence
\begin{equation}
0 \longrightarrow S(\wh{\ce}) \rmap{\wh{i}} \ind_L^U(\wh{\ce},\wh{\alpha}) \rmap{\wh{\varepsilon}} \wh{\ce} \longrightarrow  0,\label{crossedproductSES}
\end{equation}
where the maps $\wh{i}$ and $\wh{\varepsilon}$ extend $i$ and $\varepsilon$ to the crossed product algebras. We use $S$ again for the $K$-theory isomorphisms $S:K_\bullet(S(\wh{\ce}))\rightarrow K_{\bullet+1}(\wh{\ce}),\, \bullet\in\ZZ_2$.

\begin{proposition} \label{prop:d-1CT}
The $K$-groups for $\ind_L^U(\wh{\ce},\wh{\alpha})$ and $\wh{\ce} = \ce \rtimes V$
are related by the following long exact sequence
\smallskip
$$
\begin{matrix}K_1(\wh{\ce}) &\rmap{\wh{j}_*} &K_0(\ind_L^U(\wh{\ce},\wh{\alpha})) &\rmap{\wh{\varepsilon}_*} &K_0(\wh{\ce}) \cr
\lumap{1-\wh{\alpha}^{-1}_*} & & & 
&\rdmap{1-\wh{\alpha}^{-1}_*} \cr
K_1(\wh{\ce}) &\lmap{\wh{\varepsilon}_*} &K_1(\ind_L^U(\wh{\ce},\wh{\alpha})) &\lmap{\wh{j}_*} &K_0(\wh{\ce}), \cr\end{matrix}
$$
where $\wh{j}_*\coloneqq\wh{i}_*\circ S^{-1}$. Moreover, there is a commutative diagram with exact rows:
\smallskip
\small
$$
\ldots\begin{matrix}\to &K_{\bullet+1}(\ce)  &\rmap{1-\alpha^{-1}_*} &K_\bullet(S(\ce)) &\rmap{i_*} &K_\bullet(\cb) 
&\rmap{\varepsilon_*} &K_\bullet(\ce) &\to\cr
 &\ldmap{\phi_V} & &\ldmap{\phi_V} & 
 &\ldmap{\phi_V} & 
 &\ldmap{\phi_V} &\cr
\to &K_{\bullet+d}(\wh{\ce})  &\rmap{S^{-1}\circ (1-\wh{\alpha}^{-1}_*)} &K_{\bullet+d-1}(S(\wh{\ce})) &\rmap{\wh{i}_*} &K_{\bullet+d-1}(\ind_L^U(\wh{\ce},\alpha)) 
&\rmap{\wh{\varepsilon}_*} &K_{\bullet+d-1}(\wh{\ce})  &\to  \cr\end{matrix}\ldots,
$$
\normalsize
where $\phi_V: K_\bullet(*) \mapsto K_{\bullet+d-1}(* \rtimes V)$ denotes the Connes--Thom maps \cite{C} for the various actions of $V$.
\end{proposition}

\begin{proof}
The first diagram is just the long exact sequence hexagon for \eqref{crossedproductSES} with $K_\bullet(S(\wh{\ce}))$ replaced by $K_{1-\bullet}(\wh{\ce})$. The second diagram expresses the fact that the Connes--Thom  map is a natural transformation compatible with suspensions \cite{C}. 
\end{proof}

By \cite[Theorem 4.1]{PR} and Green's Imprimitivity Theorem \cite{Green, Echteroff}, the T-dual bulk algebra $\wh{\cb} = \cb\rtimes W$ can also be written as
\begin{equation}
\wh{\cb} \cong (\cb\rtimes V)\rtimes U = \ind_L^U(\wh{\ce},\alpha)\rtimes U \cong
(\wh{\ce}\rtimes_\alpha L)\otimes \cpt(\cL^2(U/L)).\label{bulkalgebrastable}
\end{equation}
Since the compact operators $\cK$ do not affect the $K$-theory, we will regard the Connes--Thom maps $\phi_U$ as isomorphisms $K_{\bullet+d-1}(\ind_L^U(\wh{\ce},\wh{\alpha}))\rightarrow K_{\bullet+d}(\wh{\ce}\rtimes_{\wh{\alpha}} L)$. We also have $\phi_W=\phi_U\circ\phi_V$. 


\begin{theorem}\label{thm:physicalhexagon}
The physical T-dual bulk algebra $\wh{\cb}$ is stably equivalent to $\wh{\ce}\rtimes_{\wh{\alpha}} L$, and the physical and geometrical $K$-groups are related by the commutative diagram:
\smallskip
$$
\ldots \begin{matrix}\to &K_{\bullet+1}(\ce)  &\rmap{1-\alpha^{-1}_*} &K_{\bullet+1}(\ce)) &\rmap{i_*} &K_\bullet(\cb) 
&\rmap{\varepsilon_*} &K_\bullet(\ce) &\to\cr
 &\ldmap{\phi_V} & &\ldmap{\phi_V} & 
 &\ldmap{\phi_W} & 
 &\ldmap{\phi_V} &\cr
\to &K_{\bullet+d}(\wh{\ce})  &\rmap{1-\wh{\alpha}^{-1}_*} &K_{\bullet+d}(\wh{\ce}) &\rmap{\phi_U\circ\wh{j}_*} &K_{\bullet+d}(\wh{\cb}) 
&\rmap{\wh{\varepsilon}_*\circ\phi_{\wh{U}}} &K_{\bullet+d-1}(\wh{\ce})  &\to  \cr\end{matrix}\ldots,
$$
and the physical Pimsner--Voiculescu (PV) exact sequence:
\smallskip
\begin{equation}
\begin{matrix}K_1(\wh{\ce}) &\rmap{\phi_U\circ\wh{j}_*} &K_1(\wh{\cb}) &\rmap{\wh{\varepsilon}_*\circ\phi_{\wh{U}}} &K_0(\wh{\ce}) \cr
\lumap{1-\wh{\alpha}^{-1}_*} & & & 
&\rdmap{1-\wh{\alpha}^{-1}_*} \cr
K_1(\wh{\ce}) &\lmap{\wh{\varepsilon}_*\circ\phi_{\wh{U}}} &K_0(\wh{\cb}) &\lmap{\phi_U\circ\wh{j}_*} &K_0(\wh{\ce}). \cr\end{matrix} \label{PVsequence1}
\end{equation}
\end{theorem}

\begin{proof}
The statement about $\wh{\cb}$ follows from \eqref{bulkalgebrastable}. The first diagram is a rewriting of the corresponding diagram in Proposition \ref{prop:d-1CT}, after simplifying the suspensions, applying $\phi_U$ (along with $\wh{\cb}\sim_{\rm stable}\wh{\ce}\rtimes_{\wh{\alpha}} L$), and noting that $\phi_{\wh{U}}$ (with $\wh{U}$ acting dually to the $U$) is the inverse of $\phi_U$ by by Takai's duality Theorem \cite{T}. The PV exact sequence follows from the lower line of the first diagram.
\end{proof}

\begin{remark}
In the case of real $C^*$-algebras, $\phi_{\wh{U}}$ should be replaced by $\phi_U^{-1}$ and the hexagons in Theorem \ref{thm:physicalhexagon} should be replaced by 24-cyclic long exact sequences.
\end{remark}


Taking account of the fact that $\wh{\cb}$ is stably isomorphic to $\wh{\ce}\rtimes L$, the sequence of 
$K$-groups in \eqref{PVsequence1} is exactly the same as that in the long exact sequence for the Toeplitz extension sequence, \cite{PV,Kellendonk1}
$$
0 \longrightarrow \wh{\ce} \longrightarrow \cT \longrightarrow \wh{\ce}\rtimes L \longrightarrow 0
$$
which was the original one used to derive the PV exact sequence. Explicitly, after identifying the $K$-theory of $\cT$ with that of $\wh{\ce}$, the Toeplitz long exact sequence is another PV sequence
\begin{equation}
\begin{matrix}K_1(\wh{\ce}) &\rmap{\jmath_*} &K_1(\wh{\ce}\rtimes L) &\rmap{\partial} &K_0(\wh{\ce}) \cr
\lumap{1-\wh{\alpha}^{-1}_*} & & & 
&\rdmap{1-\wh{\alpha}^{-1}_*} \cr
K_1(\wh{\ce}) &\lmap{\partial} &K_0(\wh{\ce}\rtimes L) &\lmap{\jmath_*} &K_0(\wh{\ce}), \cr\end{matrix} \label{PVsequence2}
\end{equation}
where $\jmath$ is inclusion into the crossed product and $\partial$ is the connecting boundary map. A natural question is whether the maps in the two hexagons \eqref{PVsequence1} (the ``mapping torus PV sequence'') and \eqref{PVsequence2} (the ``Toeplitz PV sequence'') also agree, i.e.,
\begin{equation}
    \wh{\varepsilon}_*=\partial\circ\phi_U,\qquad \phi_U\circ\wh{j}_*=\jmath_*,\label{PVintertwine}
\end{equation}
where $\phi_U$ is regarded as $K_{\bullet+d-1}(\ind_L^U(\wh{\ce},\wh{\alpha}))\rightarrow K_{\bullet+d}(\wh{\cb})\xrightarrow{\sim}K_{\bullet+d}(\wh{\ce}\rtimes_{\wh{\alpha}} L)$.

This question is related to Paschke's succinct distillation, \cite{Pas}, of a key construction in the proof of Connes' Thom Theorem, \cite{C}. Connes had initially shown, using his Thom isomorphism, that the $K$-theory of a mapping torus for a $\ZZ$-algebra $\mathcal{A}$ is isomorphic to the degree-shifted $K$-theory of the crossed product of $\mathcal{A}$ by $\ZZ$. This led to an alternative derivation of the ``mapping torus PV sequence'' without recourse to the Toeplitz long exact sequence, and is essentially what we used to obtain \eqref{PVsequence1}. On the other hand, Paschke constructs explicitly (without using Connes' Thom isomorphisms), for each $C^*$-algebra $\mathcal{A}$ with a $L=\ZZ$ action $\alpha'$, isomorphisms
$\gamma_{\alpha'}^\bullet:K_{\bullet}(\ind_L^U(\mathcal{A},\alpha'))\rightarrow K_{\bullet+1}(\mathcal{A}\rtimes_{\alpha'} L)$ which intertwine the ``Toeplitz PV sequence'' with the ``mapping torus PV sequence''. 

Replacing $\mathcal{A}\leftrightarrow\wh{\ce}$ and $\alpha'\leftrightarrow\wh{\alpha}$, the Paschke isomorphisms $\gamma_{\wh{\alpha}}^\bullet:K_\bullet(\ind_L^U(\wh{\ce},\wh{\alpha})) = K_\bullet(T_{\wh{\alpha}}(\wh{\ce}))\rightarrow K_{\bullet+1}(\wh{\ce}\rtimes_{\wh{\alpha}} L)$ are exactly such that $\wh{\varepsilon}_*=\partial\circ\gamma_{\wh{\alpha}}^\bullet$ and $\gamma_{\wh{\alpha}}^\bullet\circ\wh{j}_*=\jmath_*$. 
In the Appendix, we prove Theorem \ref{PaschkeequalsConnesThom} which says that Paschke's intertwining map is exactly the Connes--Thom isomorphism composed with the isomorphism from the imprimitivity theorem, thus our $\phi_U$, like $\gamma_{\wh{\alpha}}^\bullet$, does indeed satisfy \eqref{PVintertwine}. We also rederive Paschke's explicit formula starting from a simple axiomatic characterization of such maps.

The discussion in this section may be summarised as:

\begin{theorem} {(\bf c.f.\ \cite{C,Pas}.)}\label{thm:CTbulkboundary}
The $K$-groups for the T-dual boundary algebra $\wh{\ce} = \ce \rtimes V$ and the T-dual bulk algebra $\wh{\cb} = \cb \rtimes W\sim_{\rm stable}\wh{\ce}\rtimes_{\wh{\alpha}} L$
are related by the following Pimsner--Voiculescu sequence
\smallskip
$$
\begin{matrix}K_1(\wh{\ce}) &\rmap{\jmath_*} &K_1(\wh{\cb}) &\rmap{\partial} &K_0(\wh{\ce}) \cr
\lumap{1-\wh{\alpha}^{-1}_*} & & & 
&\rdmap{1-\wh{\alpha}^{-1}_*} \cr
K_1(\wh{\ce}) &\lmap{\partial} &K_0(\wh{\cb}) &\lmap{\jmath_*} &K_0(\wh{\ce}). \cr\end{matrix}
$$
This sequence can either be obtained as the T-dual (under the $V$ action) of the geometric long exact sequence sequence \eqref{geometricLES}, or as the sequence derived from the Toeplitz extension $$0 \longrightarrow \wh{\ce} \longrightarrow \cT \longrightarrow \wh{\ce}\rtimes L \longrightarrow 0.$$ 
\end{theorem}
In particular, we also have:
\begin{corollary}\label{cor:restrictionequalsPV}
With the assumptions on $\ce$ as in the basic setup, the following diagram commutes,
$$
\xymatrix{
K_\bullet(\cb)  \ar[d]^{\varepsilon_*} \ar[rr]^{\sim\;}_{T_W} && K_{\bullet+d}(\wh{\cb})\ar[d]^\partial \\
K_\bullet(\ce) \ar[rr]^{\sim}_{T_V} && K_{\bullet+d-1}(\wh{\ce}),} 
$$
where $T_W$ is T-duality with respect to the $W$ action on $\cb$ and $T_V$ is T-duality with respect to the $V$ action on $\ce$, implemented by the Connes--Thom isomorphisms. This shows that T-duality interchanges the geometrical restriction map with the Toeplitz bulk-boundary map.
\end{corollary}

%
%

\section{Flux and deformations}\label{sect:FD}

%
%


To connect these results with \cite[Conjecture 2.1]{HMT} we need to consider $K$-theory twisted by a Dixmier--Douady 3-form (a H-flux). There, we started with the \emph{product bundle} $B=X\times\TT^d$, on which there is a 3-form H-flux $H\in H^3(B,\ZZ)$. We will assume in this section that $X$ is locally compact second-countable Hausdorff (thus paracompact). We regard $B$ as a trivial principal torus bundle, with a fibrewise action of $\TT^d=W/N$. The $H$-twisted $K$-theory of $B$ can be identified with the $K$-theory of the associated stable complex continuous-trace algebra $\cb=\CT(B,H)$ with Dixmier--Douady invariant $H(\cb)=H$, where we recall that $\cb$ is \emph{locally} isomorphic to $C_0(B,\mathcal{K})$ but \emph{globally} twisted according to $H$. We can dimensionally reduce $H$ by the K\"{u}nneth theorem, writing $H=H_1+H_2+H_3$ where 
$$H_j\in H^j(X,H^{3-j}(\TT^d,\ZZ))=H^j(X,\ZZ)\otimes H^{3-j}(\TT^d,\ZZ),\qquad j=1,2,3,$$
thus the $j$-th component of $H$ has $j$ ``legs'' on the base $X$ and $3-j$ ``legs'' on the fibers $\TT^d$. Generally, there should also be a $H_0$ term, but we assume that it vanishes so as to avoid problems with nonassociativity when passing to the T-dual. Note that $H_3$ is supported entirely on the base, and we use the same symbol whether we regard it as living on $X$ or on a torus bundle over $X$. As explained in \cite{HM10} (see also Sect.\ \ref{sec:geometricheisenberg}), $\cb=\CT(B,H)$ has an action of $W=\RR^d$ covering the $\TT^d$ action on $B$. Its T-dual under this $W$ action is the parametrised deformation quantization $\wh{\cb}=\CT(\wh{B},H_3)_{\sigma}$, where the T-dual spectrum $\wh{B}_{H_2}$ is a $\TT^d$ bundle over $X$ with Chern class $H_2$, and $\sigma$ is the parametrised deformation built out of $H_1$.

Consider the subtorus $V/M=\TT^{d-1}\subset W/N$. Under the inclusion $\iota:E=X\times\TT^{d-1}\rightarrow X\times\TT^d=B$, there is a restriction map in twisted $K$-theory,  $\iota^*:K^{-\bullet}(B,H)\rightarrow K^{-\bullet}(E,\iota^* H)$. We can also dimensionally reduce the restricted H-flux
$$\iota^*H=(\iota^*H)_1+(\iota^*H)_2+(\iota^*H)_3, \qquad (\iota^*H)_j\in H^j(X,\ZZ)\otimes H^{3-j}(\TT^{d-1},\ZZ),$$
and it is easy to see that $(\iota^*H)_j=\iota^*H_j$. The T-dual of $\ce=\CT(E,\iota^*H)$ is then $\wh{\ce}=\CT(\wh{E},H_3)_{\iota^*\sigma}$, where $\wh{E}$ is a $\TT^{d-1}$ bundle with Chern class $\iota^*H_2$ and $\iota^*\sigma$ is the restricted deformation parameter associated to $\iota^*H_1$. With the $d$-th circle $\TT^{(d)}$ identified as $U/L$, we can identify, up to Morita equivalence, $\wh{\cb}=\CT(\wh{B},H_3)_{\sigma}$ with a crossed product of $\wh{\ce}$ by $L$ \cite{HMT}, then there is a PV boundary map $\partial:K_\bullet(\wh{\cb})\rightarrow K_{\bullet+1}(\wh{\ce})$.

The conjecture in \cite{HMT}, written in the above notation, has the condensed form of a commutative diagram
$$
\xymatrix{
K_\bullet(\cb)  \ar[d]^{\iota^*} \ar[rr]^{\sim\;}_{T} && K_{\bullet+d}(\wh{\cb})\ar[d]^\partial \\
K_\bullet(\ce) \ar[rr]^{\sim}_{T} && K_{\bullet+d-1}(\wh{\ce}).} 
$$
For the full detailed diagram, we write $\wh{B}_{H_2}\equiv \wh{B}$, $\wh{E}_{\iota^*H_2}\equiv\wh{E}$, and $q,q_a$ for the respective bundle projections. We also use $T_d, T_{d-1}$ instead of just $T$ to denote the T-duality transformations with respect to $\TT^d$ and $\TT^{d-1}$ respectively. We will now prove our conjecture, which we present as the following theorem:

\begin{theorem}\label{thm:productcase}
Let $X$ be a locally compact second countable Hausdorff space, $B=X\times\TT^d, E=X\times\TT^{d-1}$, $\iota$ be the inclusion $E\rightarrow B$, $H\in H^3(B,\ZZ)$, and let the T-dual spectra $\wh{B}, \wh{E}$ be as in the preceding paragraphs. The following diagram commutes,
\begin{equation}
\xymatrix{
K^{-\bullet}(B,\, H_1 + H_2 +H_3)  \ar[d]^{\iota^*} \ar[rr]^{\sim\;}_{T_d\;} && K_{\bullet+d}\left(\CT(\wh{B}_{H_2},\, q^*(H_3))_{\sigma}\right) \ar[d]^\partial \\
K^{-\bullet}(E,\, \iota^*H_1+\iota^*H_2 +H_3) \ar[rr]^{\sim}_{T_{d-1}} && K_{\bullet+d-1}\left(\CT(\wh{E}_{\iota^*H_2},\, q_a^*(H_3))_{\iota^*\sigma}\right)}  \label{productconjectureequation}
\end{equation}
showing that the bulk-boundary homomorphism $\partial$ is trivialised by T-duality in this parametrised context.
\end{theorem}

\begin{proof}
Let $L=\ZZ$, $M=\ZZ^{d-1}$, $N=L\oplus M$ and $U=\RR$, $V=\RR^{d-1}$, $W=U\oplus V$ as before, and $\TT=\TT^{(d)}=U/L$, $\TT^{d-1}=V/M$, $\TT^d=W/N$. By ``flux-generation'' through Corollary \ref{cor:existenceofNaction} below, there is an $N=L\oplus M$-action $\alpha\times\beta$ on an algebra $\cX$ such that $\CT(E, \iota^*H)=\Ind_M^V(\cX,\beta)$ and $\CT(B, H)=\Ind_N^W(\cX,\alpha\times\beta)$. Then $\ce\equiv\CT(E, \iota^*H)$ has the usual translation action $\tau^\beta$ by $V$ which commutes with $\alpha$, while $\cb\equiv\CT(B, H)\cong\Ind_L^U(\ce,\alpha)$; also $\cb$ has the translation action $\tau^\alpha\times\tau^\beta$ of $W$, and evaluation at 0 is the restriction $\iota^*$ to $\ce$. We are thus in the situation of Corollary \ref{cor:restrictionequalsPV}, and the commutativity of \eqref{productconjectureequation} follows.

\end{proof}

Before proving Corollary \ref{cor:existenceofNaction}, which was used in the proof of the above Theorem, we recall some general facts about the spectra and fluxes of induced algebras under spectrum-preserving actions \cite{RR}. By a result of Rosenberg, \cite[Corollary 2.2]{Ro}, \cite{EW}, a spectrum-preserving action of $N=\ZZ^d$ on a continuous-trace algebra $\cX$ with spectrum $X$ is locally inner, so for instance, $n\in N$ is implemented by $\ad(\rho(n))$ with $\rho(n)$ unitary on neighbourhoods. Furthermore, the induced algebra is again a continuous-trace algebra with spectrum $X\times\TT^d$ as a (trivial) principal $\TT^d$ bundle, with the induced $\RR^d$ translation action lifting the $\TT^d$ action (Prop.\ 3.2 of \cite{RR}). There is also a precise relationship between the Dixmier--Douady invariants (fluxes) for $\cX$, and the induced algebra given by Prop.\ 3.4 of \cite{RR}. We will only need the special case of a $\ZZ$-action, which will be iterated. 

A spectrum-preserving automorphism $\alpha\in\aut_{C_0(X)}(\cX)$ on $\cX$ is not usually globally inner; rather, Phillips and Raeburn gave an exact sequence characterising this failure when $\cX$ is stable:
$$
0 \longrightarrow \inn(\cX) \longrightarrow \aut_{C_0(X)}(\cX) \stackrel{\zeta}{\longrightarrow} H^2(X,\integer)\longrightarrow 0.
$$
Thus the class $\zeta(\alpha) \in H^2(X,\integer)$ measures the obstruction to realising $\alpha$ as a global inner automorphism $\ad (\rho)$. To find $\zeta(\alpha)$ \cite[Theorem 5.42]{RW} one chooses an open covering $\{X_j\}$ for $X$ such that for each $x\in X_j$ one can write the spectrum-preserving transformation $\alpha = \alpha(1)$ as $\ad(\rho_j(x))$. One then finds transition functions between $\rho_j$ and $\rho_k$, and proceeds exactly as in the line bundle classification to get the class $\zeta(\alpha)\in H^2(X,\integer)$. In effect, $\zeta(\alpha)$ is just the Chern class $c_1(\cL)$ of the line bundle $\cL$ whose fibre over $x\in X_j$ is just $\complex \rho_j(x)$, whose nonzero elements are precisely those implementing $\alpha$ in $X_j$, whether unitary or not. (The unitary implementers form a principal U(1)-bundle.) Furthermore, the fluxes $H(\cX)\in H^3(X,\ZZ)$ and $H(\Ind_\ZZ^\RR(\cX,\alpha))\in H^3(X\times\TT,\ZZ)$ are related by (\cite[Cor. 3.5]{RR})
\begin{equation}
H(\Ind_\ZZ^\RR(\cX,\alpha)) = \pi_X^*H(\cX) + z\cup \zeta(\alpha),\label{fluxdiscrepancy}
\end{equation}
where $\pi_X$ is the projection $X\times\TT\rightarrow X$, and $z$ is the generator of $H^1(\TT)$ (the standard volume form $d\theta$ if we use differential forms).

We will use these results to produce an \emph{arbitrary} flux $H\in H^3(X\times\TT^d,\ZZ)$ (with $H_0=0$) by inducing from a suitable $N=\ZZ^d$ action on $\cX$. This gives an explicit construction for the abstract result given in \cite{ENO} Lemma 8.1. We start with the following three basic ways to ``generate flux'' by induction.
\begin{enumerate}[{M}1:]
    \item Let $f:X\rightarrow\TT=K(\ZZ,1)$ be a continuous function representing a class $\eta\in H^1(X,\ZZ)$. Define a $\ZZ^2$ action $\beta_1$ on $C_0(X,\cK)$ as follows: the first group generator acts on the copy of $\cK=\cK(L^2(\TT))$ at $x$ by conjugation by the operator of multiplication by the identity function $\TT\rightarrow\TT$; the second generator acts by conjugation by the operator of translation by $f(x)$. These two operators commute up to a scalar, giving a projective representation of $\ZZ^2$, so their adjoint actions on $\cK$ commute. Then 
    \begin{equation}
    \Ind_{\ZZ^2}^{\RR^2}(C_0(X,\cK), \beta_1)\cong \CT(X\times\TT^2,\,\eta\cup z_1\cup z_2),\label{basicflux1}
    \end{equation}
    where $z_1,z_2$ are the generators of $H^1(\TT^{(1)},\ZZ)$ and $H^1(\TT^{(2)},\ZZ)$ respectively. This is basically the example in Section 5 of \cite{MR04a} (see also \cite{MR06}).
    \item Let $\mu\in H^2(X,\ZZ)$, and let $\beta_2\in\aut_{C_0(X)}(C_0(X,\cK))$ such that $\zeta(\beta_2)=\mu$ (this exists since $C_0(X,\cK)$ is stable and $\zeta$ is surjective). Then 
    \begin{equation}
        \Ind_\ZZ^\RR(C_0(X,\cK), \beta_2)\cong\CT(X\times\TT,\,z\cup\mu).\label{basicflux2}
    \end{equation}
    \item Let $H_3\in H^3(X)$, then the balanced tensor product $C_0(X,\cK)\otimes_X \CT(X,H_3)$ still has spectrum $X$ and flux $H_3$. This is essentially Brauer multiplication (c.f.\ Section 6.1 of \cite{RW}).
\end{enumerate}

Recall that $H^*(\TT^d,\ZZ)$ is isomorphic as a graded ring to the exterior algebra (over $\ZZ$) on $d$ generators $z_i, i=1,\ldots,d$, with $H^1(\TT^{(i)})=\ZZ[z_i]$. Thus we can further decompose each $H_j$ in $H=H_1+H_2+H_3\in H^3(X\times\TT^d,\ZZ)$ into a $\ZZ$-linear combination of elementary tensors:
\begin{align}
    H_1&=\sum_{1\leq k<l\leq d}\eta_{kl}\cup z_k\cup z_l,\qquad \eta_{kl}\in H^1(X,\ZZ), \label{cohomologydecomposition1}\\
    H_2&=\sum_{k=1}^d z_k\cup\mu_k, \qquad\qquad\qquad \mu_k\in H^2(X,\ZZ),\label{cohomologydecomposition}
\end{align}
while $H_3$ needs no further decomposition. The restrictions $\iota^*H_1, \iota^*H_2$ are given by a similar decomposition except that $1\leq k<l\leq d-1$.

\begin{proposition}[c.f..\ 8.1 of \cite{ENO}]\label{lemma:fluxgeneration}
Let $H\in H^3(X\times\TT^d,\ZZ)$ with vanishing $H_0$ component, and let $\cX$ be a stable continuous-trace algebra with spectrum $X$ and flux $H_3$. There is an action $\gamma$ of $\ZZ^d$ on $\cX$ such that $\CT(X\times\TT^d, H)\cong\Ind_{\ZZ^d}^{\RR^d}(\cX,\gamma)$.
\end{proposition}
\begin{proof}
We may assume that $\cX$ has a sufficiently large number of balanced tensor product factors $C_0(X,\cK)$ with $\cK=\cK(L^2(\TT))$, together with one factor of $\CT(X,H_3)$ contributing the $H_3$ flux component. Each flux component of $H$ in \eqref{cohomologydecomposition1}-\eqref{cohomologydecomposition} can be produced by defining a suitable action of the generators of $\ZZ^d$ on copies of $C_0(X,\cK)$ in $\cX$, using method M1 or M2 above. For example, $\eta_{kl}\cup z_k\cup z_l$ is produced by letting the $k$-th and $l$-th generators of $\ZZ^d$ act non-trivially only on one $C_0(X,\cK)$ factor as in M1; similarly $z_k\cup\mu_k$ is produced by letting the $k$-th generator act on (a different copy of) $C_0(X,\cK)$ as in M2. Compose all of these actions for each $k$ to obtain $d$ commuting $\ZZ$ actions on $\cX$. The result follows from \eqref{fluxdiscrepancy}, \eqref{basicflux1}, \eqref{basicflux2}, and the fact that inducing from $\ZZ^d$ to $\RR^d$ is the same as iteratively inducing from $\ZZ$ to $\RR$.
\end{proof}

\begin{corollary}\label{cor:existenceofNaction}
Let $H\in H^3(X\times\TT^d,\ZZ)$ with vanishing $H_0$ component, and let $\iota$ be the inclusion $X\times\TT^{d-1}\rightarrow X\times\TT^d$. There is an action $\alpha\times\beta$ of $\ZZ\oplus \ZZ^{d-1}$ on a $C^*$-algebra $\cX$, such that $\CT(X\times\TT^{d-1},\iota^*H)\cong\Ind_{\ZZ^{d-1}}^{\RR^{d-1}}(\cX,\beta)$ and $\CT(X\times\TT^d, H)\cong\Ind_{\ZZ^d}^{\RR^d}(\cX,\alpha\times\beta)$.
\end{corollary}

\begin{remark}
The concrete $\ZZ^d$ action $\gamma$ constructed in Prop.\ \ref{lemma:fluxgeneration} can be modified by inner automorphisms of $C_0(X,\cK)$ without changing the resulting flux of the induced algebra $\Ind_{\ZZ^d}^{\RR^d}(\cX,\cdot)$. If $\CT(X\times\TT^d, H)$ is given together with a particular $\RR^d$-action covering that on $X\times\TT^d$, there is a spectrum-preserving $\ZZ^d$ action $\gamma'$ on $C_0(X,\cK)=\CT(X\times\{0\},H|_{X\times\{0\}})$ by restriction, and $\Ind_{\ZZ^d}^{\RR^d}(\cX,\gamma')\cong\CT(X\times\TT^d, H)$ as $\RR^d$-algebras \cite{ENO}. The actions $\gamma'$ and $\gamma$ are exterior equivalent, so the crossed products $C_0(X,\cK)\rtimes_{\gamma'}\ZZ^d$ and $C_0(X,\cK)\rtimes_{\gamma}\ZZ^d$, which are (Morita equivalent to) the T-duals of $\Ind_{\ZZ^d}^{\RR^d}(\cX,\gamma')$ and $\Ind_{\ZZ^d}^{\RR^d}(\cX,\gamma)$, are isomorphic.
\end{remark}

%
%

\section{Examples: Heisenberg and solvable groups}\label{sect:heis}
Consider $\ce = \ind_M^V(\cX, \beta)$, where $\cX$ is a continuous-trace algebra on $X$, with spectrum preserving actions $\beta$ and $\alpha$ of $M$ and $L$ respectively; thus $\ce$ is continuous-trace with spectrum $X\times\TT^{d-1}$. We assume that that the $H_3$ flux component of $\ce$ vanishes, then $\cX$ is stably equivalent to $C_0(X,\cpt)$ (the continuous-trace algebras used in Section \ref{sect:FD} satisfy these conditions, for instance). The T-dual boundary algebra is then $\wh{\ce} = \ind_M^V(\cX, \beta)\rtimes V \cong (\cX\rtimes_\beta M)\otimes \cpt(L^2(V/M))$. The bulk algebra is 
$$
\cb = \ind_L^U(\ce,\alpha) \cong \ind_N^W(\cX, \alpha\times \beta),
$$
and, with $(\alpha\times\beta)(n) = \ad(\rho(n))$ locally, we find Mackey obstructions $\sigma=\{\sigma_x\}$ which corresponds to an $H_1$ component of flux. The T-dual bulk algebra is now
$$
\wh{\cb} \cong (\wh{\ce}\rtimes_\alpha L)\otimes \cpt 
\cong (\cX\rtimes_{\alpha\times\beta} N)\otimes \cpt,
$$
and that is Morita equivalent to the sections of a bundle of twisted group C$^*$-algebras 
$C^*(N,\sigma_x)$ over $X$ (see also Sect.\ \ref{sec:geometricheisenberg}).

The basic example of $B=\TT^3=S^1\times\TT^2$, with Dixmier--Douady invariant the volume form, was considered in detail in \cite{HMT}. In this case, we have $X = S^1$, $U \cong V \cong \real$ and $L \cong M \cong \integer$. The boundary algebra $\ce$ has spectrum $E=S^1\times\TT=\torus^2$, and there can be no Dixmier--Douady obstruction, since the dimension is less than 3, so it will be Morita equivalent to $C(\torus^2)$. The T-dual bulk algebra $\wh{\cb}$ is Morita equivalent to the group $C^*$-algebra of a discrete Heisenberg group, which can be viewed geometrically as a bundle of noncommutative tori over $S^1$. The T-dual boundary algebra $\wh{\ce}$ is the group $C^*$-algebra of a $\ZZ^2$ subgroup of the Heisenberg group, and is isomorphic to the algebra of functions on a 2-torus dual to $\torus^2$. 

{\bf Section outline:} Sect.\ \ref{sec:universalheisenberg} contains some generalities about multipliers on $\ZZ^d$ and $\RR^d$ and their relationship to (abstract) Heisenberg groups. It may be skipped to to get directly to Sect.\ \ref{sec:heisenbergtduality}, which deals with concrete generalizations of the above basic example to higher-dimensions. A new analogous example involving the solvable group is then given in Sect.\ \ref{sect:solv}.

\subsection{Universal Heisenberg groups and algebras}\label{sec:universalheisenberg}

As in \cite{HMT}, we note that the principal role of $X$ here is to \emph{parametrise} the $\mathrm{U}(1)$-valued multipliers  (or 2-cocycles) of $N$. Now the cohomology class $[\sigma]\in H^2(G,\TT)$ of a \emph{single} multiplier $\sigma$ on a separable locally compact abelian group $G$, is characterised by the antisymmetric bicharacter $\wt{\sigma}(g_1,g_2)
= \sigma(g_1,g_2)/\sigma(g_2,g_1)$, \cite{Kl,H}. Note that if $\sigma$ is already an antisymmetric bicharacter, then $\wt{\sigma}=\sigma^2$. 

For a vector group $W$ we can write the bicharacter as $\wt{\sigma} = \exp(2\pi is)$, where $s$ is a skew symmetric bilinear form, which can be thought of as a linear functional $s\in \wedge^2\wh{W}$ on the antisymmetric tensors $\wedge^2W$ such that $s(w_1\wedge w_2) = s(w_1,w_2)\in \real$; thus $H^2(W,\TT)\cong\wedge^2\wh{W}$. Each class $[\sigma]\in H^2(W,\TT)$ has a representative multiplier $\wt{\sigma}^{\frac{1}{2}}={\rm exp}(\pi i s)$ which is also an antisymmetric bicharacter, corresponding to the skew-symmetric form $(w_1,w_2)\mapsto \frac12 s(w_1, w_2)$.

The Schur multiplier (or Moore representation group) is a \lq\lq universal\rq\rq\ central extension, from which all others can be obtained, which in this case is a central extension of $W$
by $\wedge^2W$, given by the exact sequence
\begin{equation}
1 \longrightarrow \wedge^2W \longrightarrow {\rm Heis}(W) \longrightarrow W  \longrightarrow 1.\label{moorerepresentation}
\end{equation}
It is a generalised (real) Heisenberg group, Heis$(W) = W\times \wedge^2W$ whose product can be given explicitly as
$$
(w_1,p_1)\cdot(w_2,p_2) = (w_1+w_2,p_1+p_2+w_1\wedge w_2).
$$ 

Each $s\in\wedge^2\wh{W}$ determines a central extension ${\rm Heis}_s(W)$ of $W$ by $\RR$: take $W\times\real$ with product 
$$(w_1,p_1)\cdot(w_2,p_2) = (w_1+w_2,p_1+p_2+ s(w_1\wedge w_2)), \qquad W_i\in W, p_j\in\RR.$$
Let $K_s$ be the kernel of $s$. The following commutative diagram with exact rows and columns (for non-zero $s$) shows how this relates to \eqref{moorerepresentation}:
\begin{equation}
\begin{matrix}
     &   &1  &   &1   &   &   &   &   \cr
     &   &\downarrow   &   &\downarrow   &   &  &   &   \cr
 1  &\longrightarrow  &K_s  &\xlongequal{\quad}  & K_s  &\longrightarrow  &1 &  & \cr
     &   &\downarrow   &   &\downarrow   &   &\downarrow   &   &   \cr
 1  &\longrightarrow  &\wedge^2 W  &\longrightarrow  &{\rm Heis}(W)  &\longrightarrow  &W &\longrightarrow  &1\cr
     &   &\rdmap{s}   &   &\rdmap{1\times s}   &   &\|   &   &   \cr
 1  &\longrightarrow  &\real  &\longrightarrow  &{\rm Heis}_s(W)  &\longrightarrow  &W &\longrightarrow  &1\cr
     &   &\downarrow   &   &\downarrow   &   & \downarrow   &   &   \cr
     &   &1   &   &1   &   & 1  &   &   \cr
\end{matrix} \label{realheisenbergdiagram}
\end{equation}
Although $ks$ and $s$ define different group extensions, up to isomorphism ${\rm Heis}_s(W)$ is independent of which non-zero multiple of $s$ has been chosen. If we took $ks$ instead of $s$, then the map $(w,p) \mapsto (w,kp)$ gives an isomorphism of Heis$_s(W)$ and  Heis$_{ks}(W)$ as groups.

Just as for vector groups $W$, the cohomology classes $[\sigma]$ of multipliers for the lattice group $N\subset W$ are parametrised by  the antisymmetric bicharacter $\wt{\sigma}(n,n^\prime) = \sigma(n,n^\prime)/ \sigma(n^\prime,n)$, and we can write $\wt{\sigma}={\rm exp}(2\pi is)$ for some skew-symmetric $s\in\wedge^2 \wh{W}$ as before. Those $s$ which restrict to integer-valued forms on $\wedge^2 N\subset\wedge^2 W$ give trivial multipliers for $N$, and they form a lattice $(\wedge^2 N)^\perp$ in $\wedge^2 \wh{W}$, so we have $H^2(N,\TT)\cong \wedge^2\wh{W}/(\wedge^2N)^\perp\cong\wedge^2\wh{N}$.

\begin{remark}
Simply choosing $\sigma^\prime = \wt{\sigma}^{\frac12}$ as a representative multiplier for the class labelled by $\wt{\sigma}$ can lead to discontinuities when dealing with \emph{parametrised} multipliers. For example, when $m,n \in N = \integer^2$, and $X = S^1$ (identified with complex numbers unit modulus), the family of multipliers $\{\sigma_x\}_x\in S^1$ with $\sigma_x(m,n) = x^{m_1n_2}$ gives $\wt{\sigma}_x(m,n) = x^{m_1n_2-n_1m_2}$. We could take the pointwise representative multipliers to be $\sigma^\prime_x(m,n) = x^{\frac12(m_1n_2-n_1m_2)}$, but $x^{\frac12}$ is discontinuous, unless lifted to a  double cover. We give a construction of representative multipliers in Appendix \ref{sec:discretemultiplier} which avoids this problem.
\end{remark}

The Schur multiplier for $N$ is a discrete Heisenberg group Heis$(N)$, which is a central extension
\begin{equation}
1 \longrightarrow \wedge^2N \longrightarrow {\rm Heis}(N) \longrightarrow N  \longrightarrow 1,\label{universaldiscreteheisenberg}
\end{equation}
with the multiplication rule
$$
(n_1,c_1)(n_2,c_2) = (n_1+n_2,c_1+c_2+n_1\wedge n_2).
$$

The central extensions of $N$ by $\ZZ$ are classified by $H^2(N,\ZZ)\cong (\wedge^2 N)^\perp$, and a cohomology class can be labelled by a homomorphism $s:\wedge^2 N\rightarrow\ZZ$. Unlike the vector group case, $s$ need not be surjective onto $\ZZ$. For example, $s$ and $ks$ have the same kernel $K_s$ for any $0\neq k\in\ZZ$, but they have different images in $\ZZ$. This has consequences for the resulting group extensions ${\rm Heis}_s(N)$, which are defined as the central extensions
$$
1 \longrightarrow \ZZ \longrightarrow {\rm Heis}_s(N) \longrightarrow N  \longrightarrow 1,
$$
with multiplication
$$
(n_1,c_1)(n_2,c_2) = (n_1+n_2,c_1+c_2+s(n_1\wedge n_2)).
$$ 
Instead of \eqref{realheisenbergdiagram}, we have the commuting diagram
$$
\begin{matrix}
     &   &1  &   &1   &   &   &   &   \cr
     &   &\downarrow   &   &\downarrow   &   &  &   &   \cr
 1   &\longrightarrow  &K_s  &\xlongequal{\quad}  & K_s  &\longrightarrow  &1 &  & \cr
     &   &\downarrow   &   &\downarrow   &   &\downarrow   &   &   \cr
 1   &\longrightarrow  &\wedge^2 N  &\longrightarrow  &{\rm Heis}(N)  &\longrightarrow  &N &\longrightarrow  &1\cr
     &   &\rdmap{s}   &   &\rdmap{s'\coloneqq 1\times s}   &   &\|   &   &   \cr
 1   &\longrightarrow  &\ZZ  &\longrightarrow  &{\rm Heis}_s(N)  &\longrightarrow  &N &\longrightarrow  &1\cr
     &   &\downarrow   &   &\downarrow   &   & \downarrow   &   &   \cr
  1   &\longrightarrow   &\ZZ/({\rm Im}\, s)   & \xrightarrow{\,\;\sim\;}  &{\rm Heis}_s(N)/({\rm Im}\, s')    & \longrightarrow  & 1  &   &   \cr
      &   &\downarrow   &   &\downarrow   &   &   &   &   \cr
     &   &1   &   &1   &   &   &   &   \cr
\end{matrix} 
$$
As an example, take $N = \integer^d\times \integer^d$ and we write $n = (a,b) \in \integer^d\times \integer^d$. The standard integer Heisenberg group ${\rm Heis}_{s_0}(N)$ is given by $s_0((a_1,b_1),(a_2,b_2)) = a_1\cdot b_2 \in \integer$. It has a faithful representation, the {\it standard representation}, given by
\begin{equation}
\pi((a,b),c) = 
\left(\begin{matrix}
1 &a &c\cr
0 &1_d &b\cr
0 &0 &1\cr
\end{matrix} \right).\label{integerheisrep}
\end{equation} 
For a non-zero integer $k$, the central extension ${\rm Heis}_{ks_0}(N)$ is not isomorphic as a group to ${\rm Heis}_{s_0}(N)$. But if we rescale to
$$
1 \longrightarrow k^{-1}\ZZ\longrightarrow k^{-1}{\rm Heis}_{ks_0}(N) \longrightarrow k^{-1}N  \longrightarrow 1,
$$
there is an isomorphism $\psi_k:{\rm Heis}_{s_0}(N)  \to k^{-1}{\rm Heis}_{ks_0}(N)$ given by 
$\psi_k(n,c) = (k^{-1}n,k^{-1}c)$, since
\begin{eqnarray*}
\left(k^{-1}n_1,k^{-1}c_1\right)\left(k^{-1}n_2,k^{-1}c_2\right) &=& \left(k^{-1}(n_1+n_2) , k^{-1}(c_1+c_2) + ks_0(k^{-1}n_1,k^{-1}n_2)\right)\\ &=& \left(k^{-1}(n_1+n_2),k^{-1}(c_1+c_2 + s_0(n_1,n_2))\right).
\end{eqnarray*}
We can also consider ${\rm Heis}_{s_0}(k^{-1}N)$ for the rescaled lattice $k^{-1}N = \{w\in W: kw \in N\}$, which is 
$$
1 \longrightarrow k^{-2}\ZZ \longrightarrow {\rm Heis}_{s_0}(k^{-1}N) \longrightarrow k^{-1}N  \longrightarrow 1
$$
with multiplication 
\begin{eqnarray*}
\left(k^{-1}n_1,k^{-2}c_1\right)\left(k^{-1}n_2,k^{-2}c_2\right) &=& \left(k^{-1}(n_1+n_2), k^{-2}(c_1+c_2) + s_0(k^{-1}n_1,k^{-1}n_2)\right)\\
&=& \left( k^{-1}(n_1+n_2),k^{-2}(c_1+c_2 + s_0(n_1,n_2))\right),
\end{eqnarray*}
where the centre is now $k^{-2}\ZZ$ since $s_0:\wedge^2(k^{-1}N)\rightarrow k^{-2}\ZZ$. Then we also have Heis$_{s_0}(N) \cong {\rm Heis}_{s_0}(k^{-1}N)$, with isomorphism $\phi_k$ given by $\phi_k(n,c) = (k^{-1}n,k^{-2}c)$.

We have therefore proved the following: 

\begin{theorem}
With the above notation we have isomorphisms
$$
{\rm Heis}_{s_0}(k^{-1}N)\cong {\rm Heis}_{s_0}(N)  \cong k^{-1}{\rm Heis}_{ks_0}(N).
$$
\end{theorem}

From this we deduce that ${\rm Heis}_{ks_0}(N)$ is isomorphic to $k{\rm Heis}_{s_0}(k^{-1}N)$, whose elements are $k(k^{-1}n,k^{-2}c) = (n,k^{-1}c)$. This enables us to construct a matrix representation of  Heis$_{ks_0}(\integer^{2d})$  \cite{HMT}, from the standard representation $\pi$ of ${\rm Heis}_{s_0}(\integer^{2d})$: writing $n = (a,b) \in \integer^d\times \integer^d$ and $s_0((a_1,b_1),(a_2,b_2)) = a_1\cdot b_2$ as before, we obtain the faithful representation 
\begin{equation}
\pi_k((a,b),c) = 
\left(\begin{matrix}
1 &a &c/k\cr
0 &1_d &b\cr
0 &0 &1\cr
\end{matrix} \right)\label{kintegerheisrep}
\end{equation} 
of Heis$_{ks_0}(\integer^{2d})$.

With the concrete realizations \eqref{integerheisrep}-\eqref{kintegerheisrep}, it is easy to see that there is an exact sequence
$$
1\longrightarrow {\rm Heis}_{s_0}(N)\longrightarrow {\rm Heis}_{ks_0}(N)\longrightarrow \ZZ_k\longrightarrow 1.
$$

\subsubsection{Geometric picture of Heisenberg-type algebras}\label{sec:geometricheisenberg}
We can understand ${\rm Heis}(N)$ more geometrically by analysing the characters of its centre $\wedge^2 N$, which is $H^2(N,\TT)\cong \wedge^2 \wh{N}\cong\wedge^2 \TT^d$. Then $C^*({\rm Heis}(N))$ is, after stabilization, a ``universal bundle'' $\mathscr{A}$ of noncommutative tori over $\wedge^2 \TT^d$ --- the fibre over $\Theta\in\wedge^2 \TT^d$ is Morita equivalent to the noncommutative $d$-torus\footnote{$A_\Theta$ can be defined as the universal $C^*$-algebra generated by $d$ unitaries $U_1,\ldots,U_d$ subject to $U_kU_j={\rm exp}(2\pi i\Theta_{jk})U_jU_k$, where we regard $\Theta$ as a skew-symmetric $d\times d$ matrix.} $A_\Theta$.

\begin{theorem}[\cite{ENO,HM09}]
Suppose that $\cb = \ind_N^W(\cX,\ad(\nu_*))$, with $\cX = C_0(X,\cpt)$, where for each $x\in X$, $\nu_x$ is a projective representation. Regard $\nu$ as a continuous map $X\rightarrow H^2(N,\TT)$. Then the T-dual algebra $\wh{\cb}=\cb\rtimes W$ is Morita equivalent to the pullback $\nu^*(\mathscr{A})$.
\end{theorem}
Thus $\wh{\cb}$ can be visualized as a bundle of noncommutative $d$-tori over $X$. 

In particular, this picture can be used to understand $C^*({\rm Heis}_s(N))$ as follows. The homomorphism $s:\wedge^2 N\rightarrow \ZZ$ induces a homomorphism on the Pontryagin duals $\nu:\TT\rightarrow \wedge^2 \wh{N}\cong H^2(N,\TT)$. For $N = \integer^d\times \integer^d$ and surjective $s_0$ as in the previous subsection, $\nu$ is an injective degree-1 map onto its image circle in $H^2(N,\TT)\cong\wedge^2\TT^d$. Then $C^*({\rm Heis}_{s_0}(N))$ is Morita equivalent to a bundle of (stablized) noncommutative $d$-tori over $S^1=\TT$, pulled back from the universal bundle $\mathscr{A}$ under $\nu$.

For $ks_0$, we have the factorisation
$$
\begin{tikzcd}
    \wedge^2 N\arrow{rr}{ks_0}\arrow[twoheadrightarrow]{rd}{s_0} &  & \ZZ \arrow{r} & \ZZ_k \arrow{r} & 0\\
    & \ZZ\arrow[hookrightarrow]{ru}{\times k} & & &
\end{tikzcd},
$$
and dually,
$$
\begin{tikzcd}
    0\arrow{r} & \ZZ_k\arrow{r}& \TT\arrow{rr}{\nu}\arrow[twoheadrightarrow]{dr}{p} & &  H^2(N,\TT).\\
    & & & \TT\arrow[hookrightarrow]{ur} &
\end{tikzcd},
$$
so $\nu$ is a degree-$k$ map from $S^1=\TT$ onto its image circle. Then $C^*({\rm Heis}_{ks_0}(N))$ is Morita equivalent to the pullback of the bundle for $C^*({\rm Heis}_{s_0}(N))$ under the $k$-fold covering map $p$.

\subsection{Heisenberg groups, nilmanifolds, and T-duality}\label{sec:heisenbergtduality}
The standard integer (matrix) Heisenberg group $\heisZ(2d), d\in\ZZ^+$, comprises the upper triangular matrices $$\begin{pmatrix} 1 & a & c \\ 0 & 1_d & b \\ 0 & 0 & 1\end{pmatrix},\qquad a, b \in \ZZ^d, c\in \ZZ,$$
which has the multiplication law $(a_1, b_1, c_1)\cdot(a_2,b_2,c_2)=(a_1+a_2,\, b_1+b_2,\, c_1+c_2+a_1\cdot b_2)$ (this is ${\rm Heis}_{s_0}(\ZZ^{2d})$ of Sect.\ \ref{sec:universalheisenberg}). It is a subgroup of $\heisR(2d)$ defined by the same formula but with $a,b\in\RR^d, c\in\RR$. Both matrix Heisenberg groups are central extensions,
\begin{align*}
    1&\longrightarrow\RR\longrightarrow\heisR(2d)\longrightarrow\RR^{2d}\longrightarrow 1,\\
    1&\longrightarrow\ZZ\longrightarrow\heisZ(2d)\longrightarrow\ZZ^{2d}\longrightarrow 1.
\end{align*}

Let $W=\RR^{2d}$ with $w\in W$ identified with $(a,b)\in \RR^d\times\RR^d$, and let $\omega(w_1,w_2)=a_1\cdot b_2 - a_2\cdot b_1$ be the standard symplectic form on $W$. We can also define the real Heisenberg group $\heisR(W,\omega)$ as the set $W\times\RR$ equipped with the product
\begin{equation}
(w_1,p_1)\cdot(w_2,p_2) = (w_1+w_2,p_1+p_2+\frac{1}{2}\omega(w_1,w_2)), \qquad p_i\in\RR, w_i\in W, i=1,2.\label{heisenbergsymplectic}
\end{equation}
This is a special case of a more general construction given in Sect.\ \ref{sec:universalheisenberg}. The groups $\heisR(2d)$ and $\heisR(W,\omega)$ can be identified by $$\heisR(2d)\ni(a,b,c)\longleftrightarrow ((a,b),\,c-\frac{1}{2}a\cdot b)\equiv (w,p)\in\heisR(W).$$ We shall use this identification in what follows.

Take a one-dimensional subspace $U\subset W$ and $V$ a complementary subspace $\RR^{2d-1}$. We can regard $U$ as a subgroup of $\heisR(W,\omega)$ by $u\leftrightarrow (0,u)$. The restricted symplectic form $\omega_|\coloneqq\omega_{|V\times V}$ is non-degenerate only on a $2(d-1)$ dimensional subspace of $V$. The lift (preimage) $\wt{V}$ of $V$ in $\heisR(2d)$ is a normal subgroup isomorphic to $\heisR(2(d-1))\times \RR$, and we can exhibit $\heisR(2d)$ as a semi-direct product $\heisR(2d)\cong \wt{V}\rtimes U$. With $L$ and $M$ the standard lattices in $U$ and $V$ respectively, we also have $\heisZ(2d)\cong \wt{M}\rtimes L$, where $\wt{M}\cong \heisZ(2(d-1))\times\ZZ$ is the lift of $M$. In particular, we can write $C^*(\heisZ(2d))\cong C^*(\wt{M})\rtimes L$
and obtain its associated PV-boundary map 
$$\partial:K_\bullet(C^*(\heisZ(2d)))\rightarrow K_{\bullet+1}(C^*(\wt{M})).$$
The quotient $\wt{V}/\wt{M}=\nil(2(d-1))\times \TT$ is a classifying space $B\wt{M}$, and there is a fibration of $\heisR(2d)/\heisZ(2d)=\nil(2d)=B\heisZ(2d)$ over $\TT=B\ZZ$ (in fact a fibre bundle),
$$\nil(2(d-1))\times \TT \xrightarrow{\;\;\iota\;\;}\nil(2d)\longrightarrow\TT.$$
For $k\in\ZZ, k\neq 0$, we can use the symplectic form $k\omega$ instead of $\omega$, to obtain the modified integer Heisenberg groups $\heisZ(2d,k)$ and nilmanifolds $\nil_k(2d)=\heisR(2d)/\heisZ(2d,k)$.

Consider $\CT(S^1\times\TT^{2d}, d\theta\wedge k\omega)$ where $\omega$ is the standard symplectic form on $\TT^{2d}=W/N$ (lifting to $\omega$ on $W=\RR^{2d}$ as defined earlier), and $d\theta$ is the usual 1-form on $S^1$. The flux $d\theta\wedge k\omega$ is of $H_1$ type, and by the same arguments for the $d=1$ case in \cite{MR04a}, the T-dual of $\CT(S^1\times\TT^{2d}, d\theta\wedge k\omega)$ with respect to $\TT^{2d}$ is $C^*(\heisZ(2d,k))$. In more detail, we can obtain
$$\CT(S^1\times\TT^{2d}, d\theta\wedge k\omega)\cong\Ind_N^W(C(S^1,\cK),\gamma),$$
as in Corollary \ref{cor:existenceofNaction}, where above each $x\in S^1$, there is a projective representation $\nu_x$ of $N$ with multiplier $\sigma_x$, and $\gamma=\Ad(\nu_x)$. Then the T-dual is a bundle over $S^1$ of twisted group $C^*$-algebras $C^*(N,\sigma_x)$, which is what we obtain when we decompose 
$C^*(\heisZ(2d,k))$ into a twisted crossed product $C^*(\ZZ)\rtimes_\sigma N$ (see also Sect.\ \ref{sec:geometricheisenberg}). Similarly, the restriction $\CT(S^1\times\TT^{2d-1}, d\theta\wedge k\omega_|)$ where $\TT^{2d-1}=V/M$ is induced from $C(S^1,\cK)$ by $\gamma$ restricted to $M$. Its T-dual with respect to $\TT^{2d-1}$ is $(C^*(\heisZ(2(d-1,k)))\otimes C^*(\ZZ))\cong C^*(\wt{M})$ --- note that the $\omega_|$-complement of $\TT^{2d-1}$, which is a circle $\TT$, does not see any flux so it T-dualizes to the $C^*(\ZZ)$ factor. 

On the other hand, as an $S^1$ bundle over $\TT^{2d}$ with flux $d\theta\wedge k\omega$ (now of purely $H_2$ type), the T-dual of $\CT(S^1\times\TT^{2d}, d\theta\wedge k\omega)$ is $\nil_k(2d)$, a principal circle bundle over $\TT^{2d}$ with Chern class represented by $k\omega$. Similarly, the T-dual with respect to $S^1$ of $\CT(S^1\times\TT^{2d-1}, d\theta\wedge k\omega_|)$ is $\nil_k(2(d-1))\times \TT)=\wt{V}/\wt{M}$, where $\nil_k(2(d-1))$ has Chern class supported on the subtorus $\TT^{2(d-1)}$ on which $\omega_|$ is non-degenerate (this subtorus is the base of $\nil_k(2(d-1))$), while the extra $\TT$ just goes along for the ride.

The nilmanifold $\nil_k(2d)$ has an action of $\heisR(2d)$ and the symmetric imprimitivity theorem \cite{Rieffelsymmetric} gives a strong Morita equivalence
$$C(\nil_k(2d))\rtimes\heisR(2d)\sim C(\heisR(2d)\backslash \heisR(2d))\rtimes \heisZ(2d,k) = C^*(\heisZ(2d,k)),$$
which together with the generalization of the Connes--Thom isomorphism Theorem to $\heisR(2d)$ (\cite{C}, \cite{Fack81}), gives ``non-abelian T-duality'' (c.f.\ Section 3.3 of \cite{HMT})
$$K^{-\bullet}(\nil_k(2d)) =K_\bullet(C(\nil_k(2d))) \stackrel{T_{\heisR(2d)}}{\cong} K_{\bullet+1}(C^*(\heisZ(2d,k))).$$
Similarly, we have
$$K^{-\bullet}(\nil_k(2(d-1))\times\TT) \stackrel{T_{\heisR(2(d-1))\times\RR}}{\cong} K_{\bullet}(C^*(\heisZ(2(d-1),k))\otimes C^*(\ZZ)).$$

With the above definitions, we have the following commutative diagram, which is a higher-dimensional generalization of Proposition 3.3 in \cite{HMT},
$$
\small
\xymatrix{
K^{-\bullet}(S^1\times \TT^{2d},\,  d\theta \wedge \omega)  \ar[dr]_{T_1}^{\sim} \ar[dddd]^{\iota^*} \ar[rr]^{\sim\;}_{T_{2d}\;} & &  K_\bullet(C^*(\heisZ(2d))) \ar[dddd]^\partial \\
&K^{-\bullet-1} (\nil(2d)) \ar[ur]_{\;\;\;\;\;T_{\heisR(2d)}}^{\sim} \ar[dd]^{\iota^*} &\\
&&\\
&K^{-\bullet-1}(\nil(2(d-1))\times\TT)\ar[dr]^{\;\;T_{\heisR(2(d-1))\times\RR}}_{\sim} &\\
K^{-\bullet}(S^1 \times \TT^{2d-1}, d\theta\wedge\omega_|) \ar[rr]^{\sim}_{T_{2d-1}} \ar[ur]^{T_1}_{\sim} &&  K_{\bullet+1}(C^*(\heisZ(2(d-1)))\otimes C^*(\ZZ))  }  
$$
\normalsize

\begin{remark}
We can also obtain a similar commutative diagram for the generalized integer Heisenberg groups $\heisZ(2d,{\bf k})$ labelled by an $d$-tuple of integers ${\bf k}=k_1,k_2,\ldots,k_n$, as in \cite{LeePacker}. Our $\heisZ(2d,{\bf k})$ is the case where $k_1=\ldots=k_d=k$; alternatively we can consider the ${\rm Heis}_s(\ZZ^{2d})$ of Sect.\ \ref{sec:universalheisenberg} for a general $s$.
\end{remark}

\subsection{Solvable groups, solvmanifolds, and T-duality}\label{sect:solv}
Recall (cf.\ \cite{Scott}) that the real simply-connected 3 dimensional Solvable group $\solR$ is defined as a split $\RR^2$-extension of $\RR$, 
$$
 1\longrightarrow  \RR^2 \longrightarrow \solR \longrightarrow \RR \longrightarrow 1
$$
where $\RR$ acts on $\RR^2$ by 
$$
(t, (x,y)) \mapsto (e^t x, e^{-t} y), 
$$
that is, $t\in \RR$ acts on $\RR^2$ as the ${\rm SL}(2, \RR)$ matrix 
$$
\left(\begin{array}{cc} e^t & 0 \\0 & e^{-t}\end{array}\right).
$$
$\solR$ can also be defined as the matrix group,
$$
\solR = \left\{\left(\begin{array}{ccc}e^t & 0 & x\\0 & e^{-t} & y \\0 & 0 & 1\end{array}\right) \Big|\, x,y,t \in \RR\right\}.\nonumber
$$
If $\solR$ is identified with $\RR^3$ so that $(x, y)$ are the coordinates on the $\RR^2$-normal subgroup,
then the product is given by
$$
(x,y,t).(x',y',t') = (x+ e^{-t}x', y+ e^{t}y', t+t').
$$
Then $(0,0,0)$ is the identity in $\solR$. 
The inverse of $(x, y, t)$ is $(-e^t x, -e^{-t} y, -t)$, and one sees that the left invariant 1-forms on $\solR$ are
$\tau_1=e^{-t} dx$, $\tau_2= -e^t dy$ and $\tau_3=dt$.  We compute that $d\tau_1 = -dt \wedge \tau_1$, 
$d\tau_2 = dt \wedge \tau_2$, and $d\tau_3=0$.

Now a general left invariant 2-form $B$ on $\solR$ is given by 
$$
B = \frac{1}{2} \sum \Theta_{ij} \,\tau_i\wedge \tau_j =  \frac{1}{2} \tau^t \Theta \tau
$$
where $\tau$ is the column vector with $j$-th componet $\tau_j$, $\tau^t $ is the transpose of $\tau$
and $\Theta$ a skew-symmetric $(3\times 3)$ matrix.
By the computations above, $dB=0$, that is $B$ is a closed left invariant 2-form on $\solR$. A left invariant Riemannian metric is given by 
$$
ds^2 = e^{2t} dx^2 + e^{-2t} dy^2 + dt^2.
$$ 

We next recall the construction of lattices in $\solR$ and we follow \cite{ADS} closely, see also \cite{Marcolli}.
Let
$\bbK=\bbQ(d)$ be a real quadratic field and let $\iota_1, \iota_2 :\bbK\to \bbR$, be its two real embeddings. Denote by ${\bf L}\subset \bbK$ a lattice, with $U_{\bf L}^+$ the group of totally positive units preserving ${\bf L}$,
$$
U_{\bf L}^+ =\{u\in O_\bbK^* : \iota_j (u) \in \bbR_+^*, \,\text{for}\, j=1,2,   u{\bf L}\subset {\bf L}\}.
$$
Denote by $A$ a generator, so that 
$U_{\bf L}^+ =A^\ZZ = \{A^n : n \in \ZZ\}$. In the example where 
${\bf L} = O_\bbK$ is the ring of integers of $\bbK$, then the generator $A$ is a fundamental unit. 
Consider the embedding of ${\bf L}$ in $\RR^2$ given by the mapping
$$
{\bf L} \ni \ell  \longrightarrow (\iota_1(\ell), \iota_2(\ell)) \in \RR^2.
$$
Let $\Lambda = (\iota_1,\iota_2)({\bf L})$. It is a lattice in $\RR^2$, and 
 $U_{\bf L}^+$ acts on $\Lambda$ by
$$
 (\iota_1(\ell), \iota_2(\ell)) \longrightarrow  (A\iota_1(\ell), A^{-1}\iota_2(\ell)).
$$

Denote by ${\bf V}$ the group  $U_{\bf L}^+$; then ${\bf V}\cong \ZZ$. Consider the semi direct product 
$$
\solZ(A)=\Lambda \rtimes_A {\bf V}
$$
where the action of ${\bf V} = A^\ZZ$ on $\Lambda$ is induced by the action by multiplication
on ${\bf L}$. As shown in \cite{ADS}, these are all lattices in the solvable Lie group $\solR$.


The homogeneous space
$$
\solv_A = \solR/\solZ(A)
 $$ 
is a torus bundle over the circle,
$$
\TT^2 \hookrightarrow \solv_A \to \TT.
$$
The cohomology of $\solv_A$ is given by (cf. \cite{Marcolli})
$$
H^j(\solv_A) = \Bigg\{\begin{array}{l} \ZZ \quad \text{if} \,\, j=0, 3 ; \\[+7pt]
\ZZ \oplus {\rm coker}(1-A) \quad \text{if}\,\, j=2 ; \\[+7pt]
\ZZ \quad \text{if}\,\, j=1.
 \end{array} 
$$
The non-torsion generator of $H^2(\solv_A)$ has representative given by $\tau_1\wedge \tau_2$. 

The K-theory of $\solv_A$ is given by 
$$
K^j(\solv_A) = \Bigg\{\begin{array}{l} \ZZ^2  \oplus  {\rm coker}(1-A)  \quad \text{if} \,\, j=0 ; \\[+7pt]
\ZZ^2   \quad \text{if}\,\, j=1.
 \end{array} 
$$

With the above definitions, we have the following commutative diagram using Appendix A, 
$$
\small
\xymatrix{
K^{-\bullet}(\solv_A)  \ar[dd]^{\iota^*} \ar[rrr]^{\sim\;}_{T\;} & & &  K_{\bullet+1}(C^*(\solZ(A))) \ar[dd]^\partial \\
&&&\\
K^{-\bullet}(\TT^{2}) \ar[rrr]^{\sim}_{\rm id} && &  K_\bullet(C(\TT^2))  }  
$$
\normalsize
where we have used that $C(\solv_A) = \Ind_\ZZ^\RR(C(\TT^2))$ by the above, and therefore the crossed product $C(\solv_A)\rtimes \RR$ is 
strongly Morita equivalent to $C(\TT^2) \rtimes_A \ZZ \cong C^*(\solZ(A)).$ The inclusion of the fibre torus $\iota: \TT^2 \hookrightarrow \solv_A $ 
induces a restriction map $\iota^*$. Similarly, the the semi direct product 
$
\solZ(A)=\Lambda \rtimes_A {\bf V}
$
gives rise to the Pimsner-Voiculescu boundary homomorphism $\partial$, where $C^*(\Lambda)\cong C(\TT^2)$ is also used.

As with the Heisenberg groups, we may also obtain a similar commutative diagram for higher dimensional Solvable groups, but we omit the details.

\section{Bulk-boundary correspondence for higher dimensional quantum hall effect and  
topological insulators with disorder}\label{sect:disorder}

This section is concerned with the application of the abstract results in Sect.\ \ref{sect:BB}, to the bulk-boundary correspondence for the higher dimensional analogue of the quantum Hall effect and  
topological insulators, generalizing the 2D case covered in \cite{MT2}.

Let $\cA$ be a complex $C^*$-algebra with $\ZZ$-action $\alpha$. Then by Appendix \ref{appendix:paschke}, we have the commutative diagram,
\begin{equation}
    \xymatrix{K_\bullet(\Ind_\ZZ^\RR(\mathcal{A},\alpha))\ar[d]_{\iota^*}\ar[rr]^{\gamma_\alpha^\bullet} && K_{\bullet+1}(\mathcal{A}\rtimes_\alpha\ZZ)\ar[d]^{\partial}\\
    K_\bullet(\mathcal{A})\ar[rr]^{\rm Id} &&  K_{\bullet}(\mathcal{A}) 
    }\nonumber
\end{equation}
with $\gamma_\alpha^\bullet$ the Paschke isomorphism. We recall that the induced algebra $\Ind_\ZZ^\RR(\cA, \alpha) $ can be regarded as the flat fibre bundle over the circle with fibre $\cA$, that is 
$\Ind_\ZZ^\RR(\cA, \alpha) = \RR \times_{\ZZ} \cA$, and  $\iota: \cA \to \Ind_\ZZ^\RR(\cA, \alpha)$ is the inclusion of a fibre. There is a $\RR$-action on $\Ind_\ZZ^\RR(\cA, \alpha) $ by translation. Thus $\gamma_\alpha^\bullet$ implements circle T-duality in a modified sense in which the circle action on the base may not lift to the total space, but does lift to an $\RR$-action.

Now assume that $\alpha$ is a $\ZZ^d$ action on $\cA$, and write $\alpha_|$ for the restricted $\ZZ^{d-1}$-action. Then $\cA\rtimes_\alpha \ZZ^{n}$ is an $d$-fold iterated crossed product by $\ZZ$, and in particular, $(\cA\rtimes_{\alpha_|} \ZZ^{d-1})\rtimes\ZZ$. Also, $\Ind_{\ZZ^d}^{\RR^d}(\cA)\cong \Ind_\ZZ^\RR(\Ind_{\ZZ^{d-1}}^{\RR^{d-1}}(\cA))$, with the obvious actions omitted in the notation, so there is an inclusion $\iota:\Ind_{\ZZ^{d-1}}^{\RR^{d-1}}(\cA)\rightarrow \Ind_{\ZZ^d}^{\RR^d}(\cA)$. Writing $\gamma_\alpha^\bullet$ for the iterated Paschke map for the $d$-commuting $\ZZ$-actions, one has the commutative diagram
\begin{equation}
    \xymatrix{K_\bullet(\Ind_{\ZZ^d}^{\RR^d}(\cA,\alpha))\ar[d]_{\iota^*}\ar[rr]^{\gamma_\alpha^\bullet} && K_{\bullet+d}(\cA\rtimes_\alpha \ZZ^{n})\ar[d]^{\partial}\\
    K_\bullet(\Ind_{\ZZ^{d-1}}^{\RR^{d-1}}(\cA,\alpha_|))\ar[rr]^{\gamma_{\alpha_|}^\bullet} &&  K_{\bullet+d-1}(\cA\rtimes_{\alpha_|} \ZZ^{d-1}) 
    }\label{commutativityhighinduced}
\end{equation}
Alternatively, with $\ce=\Ind_{\ZZ^{d-1}}^{\RR^{d-1}}(\cA,\alpha_|)$ and $\cb=\Ind_{\ZZ^d}^{\RR^d}(\cA,\alpha)$, the hypotheses of Corollary \ref{cor:restrictionequalsPV} are satisfied and we can deduce commutativity of \eqref{commutativityhighinduced}.

Now consider a skew-symmetric matrix real $d\times d$ matrix $\Theta$ encoding a (scalar-valued) twisting 2-cocycle, 
$$(x,y)\mapsto \exp (\pi {\rm i} \langle x, \Theta (y) \rangle ),\qquad x,y\in\ZZ^d,$$
and consider the twisted crossed product algebra
$\cA\rtimes_\Theta \ZZ^{d}$. By the Packer--Raeburn stabilization trick, the adjoint map gives a canonical isomorphism
$(\cA\rtimes_\Theta \ZZ^{d})\otimes \cK \cong (\cA\otimes \cK ) \rtimes \ZZ^{d}$ with an untwisted $\ZZ^d$ action on the right-hand-side which we denote by $\alpha$. Using this in \eqref{commutativityhighinduced}, we get the commutative diagram,
\begin{equation}
    \xymatrix{K_\bullet(\Ind_{\ZZ^d}^{\RR^d}(\cA\otimes\cK,\alpha))\ar[d]_{\iota^*}\ar[rr]^{\gamma_\alpha^\bullet} && K_{\bullet+d}(\cA\rtimes_\Theta \ZZ^{d})\ar[d]^{\partial}\\
    K_\bullet(\Ind_{\ZZ^{d-1}}^{\RR^{d-1}}(\cA\otimes\cK,\alpha_|))\ar[rr]^{\gamma_{\alpha_|}^\bullet} &&  K_{\bullet+d-1}(\cA\rtimes_{{\Theta}|} \ZZ^{d-1}).
    }
\end{equation}
Let $A_\Theta$ denote the noncommutative $d$-torus and $A_{{\Theta}|}$ the  noncommutative $(d-1)$-torus, which are twisted versions of the trivial crossed products $\CC\rtimes_{\rm id}\ZZ^d$ and $\CC\rtimes_{\rm id}\ZZ^{d-1}$. Let $\CC\rightarrow \mathcal{A}$ be an equivariant homomorphism (which is automatic if it is unital), so that there is an induced morphism $A_\Theta\rightarrow \cA\rtimes_\Theta\ZZ^d$. By the naturality of the PV sequence and the Paschke maps, we have the commutative diagram
$$
\xymatrix{
K_\bullet(\Ind_{\ZZ^{d}}^{\RR^{d}}(\cA\otimes\cK))  \ar[dddd]^{\iota^*} \ar[rrr]^{\sim\;}_{T\;} & &  &K_{\bullet+d}(\cA\rtimes_\Theta \ZZ^{d}) \ar[dddd]^\partial \\
&K_\bullet(\Ind_{\ZZ^{d}}^{\RR^{d}}(\cK))\ar[r]^{\sim}_{T} \ar[dd]^{\iota^*} \ar[ul]&K_{\bullet+d}(A_\Theta) \ar[ur]\ar[dd]^{\partial} &\\
&&\\
& K_\bullet(\Ind_{\ZZ^{d-1}}^{\RR^{d-1}}(\cK))\ar[r]^{\quad\sim}_{\quad T}  \ar[dl]&K_{\bullet+d-1}(A_{{\Theta}|})\ar[dr] &\\
K_\bullet(\Ind_{\ZZ^{d-1}}^{\RR^{d-1}}(\cA\otimes\cK))\ar[rrr]^{\sim}_{T}  && & K_{\bullet+d-1}(\cA\rtimes_{{\Theta}|} \ZZ^{d-1}) }  
$$
where we have simply written $T$ for the various Paschke maps, and the diagonal maps are induced by $\CC\rightarrow\mathcal{A}$. In particular, this shows that the class in $K_0(\cA\rtimes_\Theta \ZZ^{d})$
coming from the fundamental class of the noncommutative $d$-torus  $A_\Theta$, goes to the 
class in $K_0(\cA\rtimes_{{\Theta}|} \ZZ^{d-1})$
coming from the fundamental class of the noncommutative $(d-1)$-torus  $A_{{\Theta}|}$, under the PV-boundary map. 

Specializing to the case when $\cA=C(\Sigma)$, where $\Sigma$ is a compact Cantor set with an action of $\ZZ^d$, we get precisely the higher dimensional analog of Example 4, section 4.6 in \cite{MT2}, generalizing the phenomenon that T-duality trivialises the bulk-boundary correspondence in the disordered case in 2D.

Now let $\cA$ be a \emph{real} $C^*$-algebra with a $\ZZ^d$-action $\alpha$. Since the real version of Paschke's map \cite{Pas} is also, up to a sign convention, the Connes--Thom isomorphism \cite{C} composed with an isomorphism from Green's imprimitivity theorem, it follows that the real version of the diagram above also commutes.
\begin{equation}
\xymatrix{
KO_\bullet(\Ind_{\ZZ^{d}}^{\RR^{d}}(\cA))  \ar[dddd]^{\iota^*} \ar[rrr]^{\sim\;}_{T\;} & &  &KO_{\bullet+d}(\cA\rtimes \ZZ^{d}) \ar[dddd]^\partial \\
&KO_\bullet(\Ind_{\ZZ^{d}}^{\RR^{d}}(\RR))\ar[r]^{\sim}_{T} \ar[dd]^{\iota^*} \ar[ul]&KR^{-\bullet-d}(\TT^d) \ar[ur]\ar[dd]^{\partial} &\\
&&\\
& KO_\bullet(\Ind_{\ZZ^{d-1}}^{\RR^{d-1}}(\RR))\ar[r]^{\sim}_{T}  \ar[dl]&KR^{-\bullet-d+1}(\TT^{d-1})\ar[dr] &\\
KO_\bullet(\Ind_{\ZZ^{d-1}}^{\RR^{d-1}}(\cA))\ar[rrr]^{\sim}_{T}  && & KO_{\bullet+d-1}(\cA\rtimes \ZZ^{d-1}) }  \label{realdisorderedcase}
\end{equation}
Here, $\TT^d$ is the character space of $\ZZ^d$ with involution inherited from complex conjugation, and so $KR^{-\bullet}(\TT^d)\cong KO_\bullet(\RR\rtimes_{\rm id}\ZZ^d)$.

The physical relevance of the commutative diagrams in this section is that a $\ZZ^d$ action on $\mathcal{A}=C(\Omega)$ is often used to model disorder with disorder probability space $\Omega$ (e.g.\ \cite{PSB}). The twisting by $\Theta$ is a fundamental feature of the \emph{magnetic translations} in the integer quantum Hall effect and its higher dimensional generalizations. Both disorder and twisting are crucial ingredients needed to explain phenomena like plateau in the Hall conductivity \cite{B}. In the case of real $C^*$-algebras \eqref{realdisorderedcase} says, for example, that the bulk-boundary homomorphism for time-reversal symmetric topological insulators is again trivialized into the geometric restriction map under real T-duality, even in the presence of disorder. This generalizes our earlier computations in \cite{MT2}.

%
%

\section*{Acknowledgements}

The last two authors were supported by the Australian Research Council via ARC Discovery Project grants 
DP150100008,  FL170100020
and DE170100149 respectively. The authors thank the Erwin Schr\"odinger Institute (ESI), Vienna, for its hospitality 
during the ESI Program on 
{\em Higher Structures in String Theory 
and Quantum Field Theory},
when part of this research was completed.

%
%

\appendix
\section{Appendix: Paschke's map and the Connes--Thom isomorphism}\label{appendix:paschke}

This appendix explains why Paschke's map \cite{Pas} is, up to a sign convention, the Connes--Thom isomorphism \cite{C} composed with an isomorphism from Green's imprimitivity theorem. This result belongs to a family of ideas which can be found in \cite{C,Cun,CMR,Pas} and may be known to experts, but we provide a detailed argument for the reader's reference. The Paschke map is explicitly defined for each $\ZZ$-algebra $(\mathcal{A},\alpha)$, whereas the Connes--Thom map is more abstractly defined. As we will see, there is an analogous abstract characterization of the Paschke map, which determines its formula uniquely.

\subsection*{Generalities on Connes--Thom isomorphism}
Recall from \cite{C} that the Connes--Thom isomorphism is a natural transformation between the functors $K_\bullet(\cdot)$ and $K_{\bullet+1}((\cdot)\rtimes\RR)$, $\bullet\in\ZZ_2$, i.e.\ it assigns to every $\RR$-$C^*$-algebra $(\mathscr{A},\alpha)$ an isomorphism $\phi^\bullet_\alpha: K_\bullet(\mathscr{A})\rightarrow K_{\bullet+1}(\mathscr{A}\rtimes_\alpha\RR)$, such that for any morphism $\rho:(\mathscr{A},\alpha)\rightarrow(\mathscr{B},\beta)$, the following diagram commutes:
\begin{equation}
    \xymatrix{K_\bullet(\mathscr{A})\ar[d]_{\rho_*}\ar[rr]^{\phi_\alpha^\bullet} && K_{\bullet+1}(\mathscr{A}\rtimes_\alpha\RR)\ar[d]^{(\rho\rtimes\RR)_*}\\
    K_\bullet(\mathscr{B})\ar[rr]^{\phi_\beta^\bullet} &&  K_{\bullet+1}(\mathscr{B}\rtimes_\beta\RR) 
    }
\end{equation}
Here, $\rho\rtimes\RR$ is the morphism of crossed products induced by the equivariant map $\rho$. Furthermore, Connes shows that $\phi_\alpha^\bullet$ is an isomorphism which is determined uniquely by this naturality property together with a normalization condition (for $(\mathcal{A},\alpha)=(\CC,\mathrm{id})$) and compatibility with suspensions.

\subsection*{The analogous axioms for Paschke's map}
If $\rho:(\mathcal{A},\alpha)\rightarrow (\mathcal{B},\beta)$ is a morphism of $\ZZ$-$C^*$-algebras (not necessarily unital), the induction functor $\Ind\coloneqq\Ind_\ZZ^\RR(\cdot)$ gives a morphism of induced algebras $\tilde{\rho}:\Ind(\mathcal{A},\alpha)\rightarrow\Ind(\mathcal{B},\beta)$ which is equivariant for the respective translation actions $\tau^\alpha, \tau^\beta$ of $\RR$. Elements of $\Ind(\mathcal{A},\alpha)$ are viewed either as bounded continuous functions $f:\RR\rightarrow \mathcal{A}$ satisfying an equivariance condition\footnote{We have switched convention for induced algebras in the Appendix compared to the main text. This is to make closer contact to Paschke's original work, and the effect is that $\alpha^{-1}$ is replaced by $\alpha$.}, $f(x+1)=\alpha(f(x)), x\in\RR$, or alternatively, continuous $f:[0,1]\rightarrow\mathcal{A}$ such that $f(1)=\alpha(f(0))$ (the mapping torus). The translation action $\tau^\alpha$ of $\RR$ on $\Ind(\mathcal{A},\alpha)$ is $(\tau^\alpha_t f)(s)=f(s+t), s,t\in\RR$.

Let $S\alpha$ be the suspended $\ZZ$-action on $S\mathcal{A}$ (acting by $\alpha$ on $\mathcal{A}$ and trivially on the suspension variable). Note that $S(\Ind(\mathcal{A},\alpha))\cong \Ind(S\mathcal{A},S\alpha)$ and $S(\mathcal{A}\rtimes_\alpha\ZZ)\cong S\mathcal{A}\rtimes_{S\alpha}\ZZ$, and write $S^\bullet_{(\cdot)}$ for the natural suspension isomorphisms $K_\bullet(\cdot)\rightarrow K_{\bullet-1}(S(\cdot))$.

For $\bullet\in\ZZ_2$, let $\gamma^\bullet$ be a natural transformation of the functors $(\mathcal{A},\alpha)\mapsto K_\bullet(\Ind(\mathcal{A},\alpha))$ and $(\mathcal{A},\alpha)\mapsto K_{\bullet+1}(\mathcal{A}\rtimes_\alpha\ZZ)$ from $\ZZ$-algebras to abelian groups, satisfying the following three axioms:
\begin{axioms}[Normalization]\label{normalization}
If $(\mathcal{A},\alpha)=(\CC,\mathrm{id})$, then $\gamma^0_{\mathrm{id}}:K_0(C(\RR/\ZZ))=K_0(C(\TT))\rightarrow K_1(C^*(\ZZ))$ takes $[1_{C(\TT)}]$ to the (Bott) generator $[b]$ of $K_1(C^*(\ZZ))$ corresponding to $1\in\ZZ$ regarded as an element of $C^*(\ZZ)$.
\end{axioms}
\begin{axioms}[Naturality]\label{naturality}
If $\rho:(\mathcal{A},\alpha)\rightarrow(\mathcal{B},\beta)$ is a morphism of $\ZZ$-algebras, then the following diagram commutes:
\begin{equation}
    \xymatrix{K_\bullet(\Ind(\mathcal{A},\alpha))\ar[d]_{\tilde{\rho}_*}\ar[rr]^{\gamma_\alpha^\bullet} && K_{\bullet+1}(\mathcal{A}\rtimes_\alpha\ZZ)\ar[d]^{(\rho\rtimes\ZZ)_*}\\
    K_\bullet(\Ind(\mathcal{B},\beta))\ar[rr]^{\gamma_\beta^\bullet} &&  K_{\bullet+1}(\mathcal{B}\rtimes_\beta\ZZ) 
    }
\end{equation}
\end{axioms}
\begin{axioms}[Suspension]\label{suspension}
$\gamma$ is compatible with $S$ in the sense that
\begin{equation}
    S^{(\bullet+1)}_{\mathcal{A\rtimes_\alpha\ZZ}}\circ\gamma^\bullet_\alpha=\gamma^{(\bullet-1)}_{S\alpha}\circ S^\bullet_{\Ind(\mathcal{A},\alpha)}.
\end{equation}
\end{axioms}

\begin{proposition}\label{uniqueness}
There is a unique $\gamma$ satisfying Axioms \ref{normalization}-\ref{suspension}. 
\end{proposition}

\begin{proof}
{\bf Outline:} We use a modification of Connes' argument for the uniqueness of $\phi$. Axiom \ref{suspension} and Bott periodicity ensures that we only need to look at $\gamma^0$ (which we will simply denote by $\gamma$ subsequently). We need to show that $\gamma_\alpha[p]\in K_1(\mathcal{A}\rtimes_\alpha\ZZ)$ is uniquely determined for any projection $p\equiv\{p_\theta\}_{\theta\in[0,1]}$ in $\Ind(\mathcal{A},\alpha)$. The basic idea is to construct a $\ZZ$-action $\alpha'$, exterior equivalent to $\alpha$, which fixes $p_0$. Furthermore the new induced algebra $\Ind(\mathcal{A},\alpha')$, with its translation action $\tau^{\alpha'}$, is $\RR$-equivariantly isomorphic to the original $\Ind_{\alpha}\mathcal{A}$ with a \emph{modified} action $\tau'$, where $\tau'$ is exterior equivalent to the original $\tau^\alpha$ and fixes $p$. Thus the projection $p'\in\Ind(\mathcal{A},\alpha')$ corresponding to $p\in\Ind(\mathcal{A},\alpha)$ under this isomorphism is itself translation invariant under $\tau^{\alpha'}$. Then the Axioms determine what $\gamma_{\alpha'}[p']$ has to be. Finally we need to argue that the modification $\alpha\mapsto\alpha'$ can be assumed because the above exterior equivalences determine unique maps, consistent with the Axioms, which yield the desired $\gamma_\alpha[p]$ from $\gamma_{\alpha'}[p']$.

{\bf Details:} 
As in \cite{C}, we may assume without loss that $\mathcal{A}$ is unital and $p\in\Ind(\mathcal{A},\alpha)$. By \cite{C} Prop.\ 4 (c.f.\ \cite{CMR} Lemma 10.16, \cite{Rieffel2}), we may also assume that on $\Ind(\mathcal{A},\alpha)$ there is an action $\RR$-action $\tau'$, exterior equivalent to $\tau^\alpha$, which fixes $p$. Let $U\equiv\{U_t\}_{t\in\RR}$ be the 1-cocycle on $\Ind(\mathcal{A},\alpha)$ which relates $\tau^\alpha$ and $\tau'$ via $\tau'_t=\Ad(U_t)\circ\tau^\alpha_t$. Thus \ $\{U_t\}_{t\in\RR}$ satisfies the cocycle condition and equivariance condition, 
\begin{equation}
    U_{t_1+t_2}=U_{t_1}\tau^\alpha_{t_2}(U_{t_2}),\qquad U_t(1)=\alpha(U_t(0)), \qquad t,t_1,t_2\in\RR.\label{1cocycleinduced}
\end{equation}
Define the modified $\ZZ$-action $\alpha'$ on $\mathcal{A}$ via 
\begin{equation}
    \alpha'_n=\Ad(U_n(0))\circ\alpha_n, \qquad n\in\ZZ.\label{modifiedZaction}
\end{equation}
We can verify, from \eqref{1cocycleinduced}, that the assignment $u:n\mapsto U_n(0)\eqqcolon u_n,\, n\in\ZZ$ defines a unitary 1-cocycle $\{u_n\}_{n\in\ZZ}$ on $(\mathcal{A},\alpha)$, i.e., $u_{m+n}=u_m\alpha^m(u_n)$. Thus $\mathcal{A}\rtimes_{\alpha'}\ZZ\cong\mathcal{A}\rtimes_\alpha\ZZ$ (a canonical isomorphism $\varphi_u$ is given in Lemma \ref{inner1} later).

Since $\tau'_1=\Ad(U_1)\circ \tau^\alpha_1$ fixes $p$, it follows that $p_0\coloneqq p(0)\in\mathcal{A}$ fixed by $\alpha'$. Define the isomorphism $\Psi_U:\Ind(\mathcal{A},\alpha)\rightarrow\Ind(\mathcal{A},\alpha')$ by 
\begin{equation}
   \Psi_Uf(s)=\Ad(U_s(0))[f(s)],\qquad s\in\RR, f\in \Ind(\mathcal{A},\alpha).\label{isoZcrossed}
\end{equation}
We can verify that $\Psi_Uf$ does satisfy $\alpha'$-equivariance,
and intertwines $\tau'$ on $\Ind(\mathcal{A},\alpha)$ with the new translation action $\tau^{\alpha'}$ on $\Ind(\mathcal{A},\alpha')$, i.e.\ $(\tau^{\alpha'}_t(\Psi_U f))= (\Psi_U(\tau'_t f))$.

Observe that $p'(\theta)=p'_0=p_0$ since $p'=\Psi_U(p)$ is fixed by $\tau^{\alpha'}$. The homomorphism
\begin{equation*}
    \omega:(\CC,\mathrm{id})\rightarrow (\mathcal{A},\alpha'),\;\;\;\;\;\omega(\lambda)=\lambda p'_0
\end{equation*}
is $\ZZ$-equivariant, and we write $\tilde{\omega}$ for the corresponding inflated map $\tilde{\omega}:C(\TT)\cong\Ind(\CC,\mathrm{id})\rightarrow\Ind(\mathcal{A},\alpha')$. Note that
\begin{equation*}
    (\tilde{\omega} (1_{C(\TT)}))(\theta)\equiv \omega(1_{C(\TT)}(\theta))=p'_0=p'(\theta),\qquad \theta\in [0,1],
\end{equation*}
so $p'=\tilde{\omega}(1_{C(\TT)})$. Then naturality and normalization means that $\gamma_{\alpha'}[p']$ must be determined by the equation
\begin{equation}
    \gamma_{\alpha'}[p']=\gamma_{\alpha'}\circ\tilde{\omega}_*[1_{C(\TT)}]=(\omega\rtimes\ZZ)_*\circ \gamma_{\mathrm{id}}[1_{C(\TT)}]=(\omega\rtimes\ZZ)_*[b].\label{CTforfixedprojection}
\end{equation}
Corollary \ref{exteriorreplacement} gives the details of how $\gamma[p]$ is obtained from $\gamma_{\alpha'}[p']$. 
\end{proof}

\subsubsection*{Connes' $2\times 2$ matrix trick}
\begin{lemma}[c.f.\ \cite{C} Lemma 2]\label{inner1}
Let $\alpha, \alpha'$ be exterior equivalent $\ZZ$-actions on $\mathcal{A}$ related by a unitary 1-cocycle $\ZZ\ni n \mapsto u_n \in U\mathcal{M}\mathcal{A}$. Then 
\begin{equation}
    \kappa\begin{pmatrix}a_{11} & a_{12} \\ a_{21} & a_{22}
    \end{pmatrix}= \begin{pmatrix} \alpha(x_{11}) & \alpha(x_{12})u_1^* \\ u_1\alpha(x_{21}) & \alpha'(x_{22}) 
    \end{pmatrix}
\end{equation}
is a $\ZZ$-action on $M_2(\mathcal{A})$ which restricts to $\alpha$ in the top-left corner and to $\alpha'$ in the bottom-right corner. Let $\iota, \iota'$ be the equivariant inclusions of $\mathcal{A}$ into their respective corners in $M_2(\mathcal{A})$. There is a unique isomorphism $\varphi_u:\mathcal{A}\rtimes_\alpha\ZZ\rightarrow\mathcal{A}\rtimes_{\alpha'}\ZZ$ such that $(\iota'\rtimes\ZZ)(\varphi_u(y))=\Ad(\begin{pmatrix} 0 & 1 \\ 1 & 0\end{pmatrix})[(\iota\rtimes\ZZ)(y)]$ for all $y\in\mathcal{A}\rtimes_\alpha\ZZ$.
\end{lemma}

Similarly, there are inclusions $\tilde{\iota}, \tilde{\iota'}$ of the induced algebras into the respective corners of $\Ind_\kappa M_2(\mathcal{A})$. Let $\tau^\alpha, \tau'$ be the exterior equivalent $\RR$-actions on $\Ind(\mathcal{A},\alpha)$ as in the proof of Proposition \ref{uniqueness}, related by the 1-cocycle $\{U_t\}_{t\in\RR}$. A straightforward computation shows:
\begin{lemma}\label{inner2}
Define $X\in\Ind_{\kappa} M_2(\mathcal{A})$ by
\begin{equation*}
X(s)=\begin{pmatrix} 0 & U_s^*(0)\\ U_s(0) & 0
\end{pmatrix}.
\end{equation*}
Then the isomorphism $\Psi_U$ satisfies $\tilde{\iota'}(\Psi_U(f))=\Ad(X)[\tilde{\iota}f], f\in\Ind(\mathcal{A},\alpha)$.
\end{lemma}

\begin{corollary}\label{exteriorreplacement}
Suppose $\gamma$ satisfies Axioms \ref{normalization}-\ref{suspension}, and $\alpha, \alpha', U$ are as above. Then $\gamma^\bullet_{\alpha'}=(\varphi_u)_*\circ \gamma^\bullet_\alpha\circ(\Psi_U^{-1})_*$.
\end{corollary}

\begin{proof}
Recall that inner automorphisms induce the identity map in $K$-theory. Thus Lemmas \ref{inner1} and \ref{inner2} show that $(\iota\rtimes\ZZ)_*=(\iota'\rtimes\ZZ)_*\circ(\varphi_u)_*$ and $\tilde{\iota}_*=\tilde{\iota'}_*\circ(\Psi_U)_*$. Using the naturality axiom, and dropping the superscript ${}^\bullet$ for now,
\begin{equation*}
    (\iota'\rtimes\ZZ)_*\circ\gamma_{\alpha'}= \gamma_\kappa\circ \tilde{\iota'}_*= \gamma_\kappa\circ \tilde{\iota}_*\circ(\Psi_U^{-1})_*= (\iota\rtimes\ZZ)_*\circ\gamma_\alpha\circ(\Psi_U^{-1})_*= (\iota'\rtimes\ZZ)_*\circ(\varphi_u)_*\circ\gamma_\alpha\circ(\Psi_U^{-1})_*,
\end{equation*}
so injectivity of $(\iota'\rtimes\ZZ)_*$ (see \cite{C} Proposition 2-3) gives the required result.
\end{proof}

\subsection*{Paschke's isomorphisms}
Paschke's isomorphisms \cite{Pas}
$$\gamma^{\bullet,\rm{Paschke}}_\alpha: K_\bullet(\Ind(\mathcal{A},\alpha))\rightarrow K_{\bullet+1}(\mathcal{A}\rtimes_\alpha\ZZ)$$ 
are first defined for unital $\mathcal{A}$ and $\bullet=0$ on $[p]$, by exhibiting the existence of a path of unitaries $\theta\mapsto w_\theta\in \mathcal{A}$ such that 
\begin{equation}
    p(\theta)=\Ad(w_\theta)(p_0),\qquad \theta\in[0,1].\label{Paschkeunitary}
\end{equation}
In particular, $p_1=\Ad(w_1)(p_0)$. Then $\gamma^{0,\rm{Paschke}}_\alpha[p]$ is defined to be \begin{equation}
    \gamma^{0,\rm{Paschke}}_\alpha[p]=[L^*w_1p_0+1-p_0],\label{Paschkeformula}
\end{equation}
where $L\in\mathcal{M}(\mathcal{A}\rtimes_\alpha\ZZ)$ is the unitary implementing $\alpha$ (i.e.\ $\alpha(a)=LaL^*, a\in\mathcal{A}$), with this map shown to be well-defined in $K$-theory\footnote{There is no loss of generality in assuming that $\mathcal{A}$ is unital and that $p$ is in $\Ind(\mathcal{A},\alpha)$ rather than in $\Ind(\mathcal{A},\alpha)\otimes M_n$, see \cite{Pas}.}.

The $\bullet=1$ case is defined by compatibility with suspensions. Since $\gamma^{0,\rm{Paschke}}_\alpha$ may be expressed using representative projections and unitaries (rather than their $K$-theory classes), it is easy to see that the naturality axiom \ref{naturality} is satisfied. It is also clear that normalization axiom \ref{normalization} is satisfied up to a minus sign.

On the other hand, there is a commutative diagram due to the Connes--Thom natural isomorphisms,
\begin{equation}
    \xymatrix{K_\bullet(\Ind(\mathcal{A},\alpha))\ar[d]_{\tilde{\rho}_*}\ar[rr]^{\phi_{\tau^\alpha}^\bullet} && K_{\bullet+1}(\Ind(\mathcal{A},\alpha)\rtimes_\alpha\RR)\ar[d]^{(\tilde{\rho}\rtimes\RR)_*}\\
    K_\bullet(\Ind(\mathcal{B},\beta))\ar[rr]^{\phi_{\tau^\beta}^\bullet} &&  K_{\bullet+1}(\Ind(\mathcal{B},\beta)\rtimes_\beta\RR) 
    }.
\end{equation}
Composing $\phi_{\tau^\alpha}^\bullet$ with the natural $K$-theory isomorphisms 
\begin{equation*}
    M_\alpha^{\bullet+1,\rm Green}:K_{\bullet+1}(\Ind(\mathcal{A},\alpha)\rtimes_{\tau^\alpha}\RR)\cong K_{\bullet+1}(\mathcal{A}\rtimes_\alpha \ZZ)
\end{equation*}
given by the Morita equivalence $\Ind(\mathcal{A},\alpha)\rtimes_{\tau^\alpha}\RR\sim_{\rm M.E.}\mathcal{A}\rtimes_\alpha\ZZ$ (this is a special case of Green's imprimitivity theorem \cite{Green}), we obtain another $\gamma$ satisfying Axioms \ref{normalization}-\ref{suspension}: The normalization axiom follows from the fact that $\phi_{\tau^{\rm id}}$ is an isomorphism, so there is only one possibility (apart from a sign choice) for $M_{\rm id}^{0,\rm Green}\circ\phi_{\tau^{\rm id}}:\ZZ[1_{C(\TT)}]\rightarrow\ZZ[b]$. Compatibility with suspensions is inherited from the Connes--Thom map, while naturality follows from naturality of implementing the Morita equivalence (see \cite{Echteroff} Chapter 4). By Proposition \ref{uniqueness}, we have

\begin{theorem}\label{PaschkeequalsConnesThom}
The map $\gamma^{\bullet,\,\rm{Paschke}}_\alpha:K_\bullet(\Ind(\mathcal{A},\alpha))\rightarrow K_{\bullet+1}(\mathcal{A}\rtimes_\alpha\ZZ)$ is, up to a sign, the composition $M_\alpha^{\bullet+1,\,\rm Green}\circ\phi^\bullet_{\tau^\alpha}$.

\end{theorem}

\subsubsection*{Paschke's unitaries and Connes' cocycle}
Given a $p\in\Ind(\mathcal{A},\alpha)$, we may start from Connes' cocycle $\{U_t\}_{t\in\RR}$ in \eqref{1cocycleinduced} above, and obtain a path of unitaries $w_\theta\coloneqq U^*_\theta(0), \theta\in[0,1]$ in $\mathcal{A}$. Then $w_0=1$ and
\begin{align*}
    p_\theta=p(\theta)=\tau^\alpha_\theta(p(0))=(\Ad(U^*_\theta)\circ\tau'(p))(0)=\Ad(U^*_\theta(0))(p(0))=\Ad(w_\theta)(p_0),
\end{align*}
thus $\{w_\theta\}_{\theta\in[0,1]}$ is precisely a Paschke path whose existence was proved in \cite{Pas}.

Conversely, given a Paschke path $\{w_\theta\}_{\theta\in[0,1]}$ such that $w_0=1$ and $\Ad(w_\theta)(p_0)=p_\theta$, we can reconstruct Connes' cocycle for the action $\tau^\alpha$ on $\Ind(\mathcal{A},\alpha)$ as follows. Regard $p\in\Ind(\mathcal{A},\alpha)$ as a projection-valued function $\RR\rightarrow\mathcal{A}$ through the equivariance condition $p(s)\equiv p_s=\alpha(p_{s-1}), s\in\RR$. Similarly, extend Paschke's $w:[1,0]\mapsto \mathcal{A}$ to a continuous unitary function $w:\RR\ni s\mapsto w_s\in\mathcal{A}$ by the recursion relation $w_s=\alpha(w_{s-1})\cdot w_1$. It is easy to check that
$p_s=\Ad(w_s)(p_0)$ for all $s\in\RR$, and that $n\mapsto w_n$ defines a 1-cocycle for $(\mathcal{A},\alpha)$. The projection $p_0$ is then fixed by $\alpha'_n=\Ad(w^*_n)\circ\alpha, n\in\ZZ$, since $$\alpha'(p_0)=\Ad(w^*_1)(\alpha(p_0))=\Ad(w^*_1)\circ\Ad(w_1)(p_0)=p_0.$$

The assignment $t\mapsto U_t\in\Ind(\mathcal{A},\alpha)$ defined by
\begin{equation}
U_t(s)=w_sw_{s+t}^*, \qquad s,t\in\RR\label{Connescocycle}
\end{equation}
is the desired Connes 1-cocycle for $(\Ind(\mathcal{A},\alpha),\tau^\alpha)$; we can verify that each $U_t$ is $\alpha$-equivariant and satisfies the cocycle condition. If we define the $\RR$-action $\tau'$ on $\Ind(\mathcal{A},\alpha)$ by $\tau'_t=\Ad(U_t)\circ(\tau^\alpha_t)$, we find that the projection $p$ is now fixed by $\tau'$.

\subsection*{Paschke's formula from axioms}
Recall that $\gamma^0_\alpha[p]$ can be computed by passing to $\alpha'$, using $\gamma^0_\alpha[p]=(\varphi_u^{-1})_*\circ\gamma^0_{\alpha'}\circ(\Psi_U)_*[p]$.
In \eqref{CTforfixedprojection}, we saw that $\gamma^0_{\alpha'}\circ(\Psi_U)_*[p]=\gamma^0_{\alpha'}[p']=(\omega\rtimes\ZZ)_*[b].$
Note that even if $\mathcal{A}$ is unital, the homomorphism $\omega$ (and thus $\omega\rtimes\ZZ$) is \emph{non-unital} in general, so the induced map $(\omega\rtimes\ZZ)_*$ in $K_1$ is slightly tricky to compute (c.f. Prop.\ 8.1.6 of \cite{Rordam}). Namely, we have to append a formal unit ${\bf 1}$ to $C(\TT)$ and $\mathcal{A}\rtimes_{\alpha'}\ZZ$, replace $[b]$ by $[{\bf 1}-1_{C(\TT)}+b]$, then compute $(\omega\rtimes\ZZ)_*$ using the unital extension $(\omega\rtimes\ZZ)^+:C(\TT)^+\rightarrow (\mathcal{A}\rtimes_{\alpha'}\ZZ)^+$. This gives
\begin{align*}
((\omega\rtimes\ZZ)^+)_*[{\bf 1}-1_{C(\TT)}+b]&=[{\bf 1}+(\omega\rtimes\ZZ)(-1_{C(\TT)}+b)]\\&=[{\bf 1}-L'p'_0+p'_0]\\
&=[{\bf 1}-1_{\mathcal{A}\rtimes_{\alpha'}\ZZ}+1_{\mathcal{A}\rtimes_{\alpha'}\ZZ}+L'p'_0-p'_0],
\end{align*}
and mapping back to unitaries in $\Ind(\mathcal{A},\alpha')$, we obtain
$$(\omega\rtimes\ZZ)_*[b]=[1+L'p'_0-p'_0]=[L'p_0+(1-p_0)].$$ Finally, recall that the isomorphism $\varphi_u$ implements the isomorphism $\mathcal{A}\rtimes_\alpha\ZZ\cong\mathcal{A}\rtimes_{\alpha'}\ZZ$ induced by the exterior equivalence $\alpha'_n=\Ad(w^*_n)\circ\alpha_n$, so it effects $w_1^*L\mapsto L'$. Thus,
\begin{equation*}
    \gamma^0_\alpha[p]=(\varphi_u^{-1})_*\circ\gamma^0_{\alpha'}\circ(\Psi_U)_*[p]=(\varphi_u^{-1})_*[L'p_0+(1-p_0)]=[w_1^*Lp_0+(1-p_0)],
\end{equation*}
which, up to a sign, concides with $\gamma^{0,\rm Paschke}_\alpha[p]$.

\subsection*{The real $C^*$-algebra case}
A general reference for the $K$-theory of real $C^*$-algebras is \cite{Schroder}. The Connes--Thom isomorphisms and Pimsner--Voiculescu exact sequence also hold for real $C^*$-algebras $\mathcal{A}$ and real crossed products, see \cite{Ro2, Schroder}. Similarly, Paschke's map $\gamma^{0,\rm Paschke}_\alpha$ still makes sense on real $\ZZ$-algebras $(\mathcal{A},\alpha)$. We can then repeat the arguments in this appendix to obtain the real version of Theorem \ref{PaschkeequalsConnesThom}. The only significant change is in the normalization axiom, and in taking $\bullet\in\ZZ_8$ rather than in $\ZZ_2$.

Recall that $\mathbb{C}\rtimes_{\rm id}\ZZ=C_\RR^*(\ZZ)\cong C(\TT;\varsigma)$ where $\varsigma$ is the involution $\theta\mapsto -k\theta$ on $\TT=\widehat{\ZZ}$ inherited from complex conjugation of characters (parametrised by $\theta\in[0,1]$), and $C(\TT;\varsigma)$ is the \emph{real} $C^*$-algebra of continuous functions $f\in\TT\rightarrow \CC$ such that $\overline{f(\theta)}=f(-\theta)$. These functions are precisely the ones with real Fourier coefficients. It is known (pp.\ 40 of \cite{Schroder}) that $K_1(C(\TT;\varsigma))\cong\ZZ[z]\oplus\ZZ_2[-1_{C(\TT,\varsigma)}]$, where $z$ denotes the unitary function $\theta\mapsto e^{2\pi i\theta}$ with winding number 1. Furthermore, $[z]$ is the Bott element implementing (1,1) periodicity in $KR$-theory (c.f.\ Theorem 1.5.4 of \cite{Schroder}, Theorem 10.3 of \cite{LM}), so we also write it as $[b_r]$. We can also identify it with the unitary $L$ implementing the (trivial) automorphism in $\RR\rtimes_{\rm id}\ZZ$.

On the other hand, the induced algebra for $(\RR,{\rm id})$ is just the continuous real-valued functions on the circle $C(\TT;\RR)$. It is also known that $$KO_0(C(\TT;\RR))=\ZZ[1_{C(\TT;\RR)}]\oplus\ZZ_2\left[[P_{\rm Mob}]-[1_{C(\TT;\RR)}]\right],$$ where $P_{\rm Mob}$ is the M\"{o}bius projection in $M_2(C(\TT;\RR))$ given by
\begin{equation*}
    P_{\rm Mob}(\theta)=\Ad(R(\theta))\begin{pmatrix} 1 & 0 \\ 0 & 0    \end{pmatrix}=\Ad\begin{pmatrix}\cos\,\pi\theta & -\sin \,\pi\theta \\ \sin\,\pi\theta & \cos\,\pi\theta\end{pmatrix}\begin{pmatrix} 1 & 0 \\ 0 & 0    \end{pmatrix}.
\end{equation*}

It is straightforward to compute the real Pashcke map on $KO_0(C(\TT;\RR))$. For the constant projection $1$, the Paschke path of unitaries is trivial, so we have $\gamma^{0,\rm Paschke}_{\rm id}[1]=[L+1-1]=[L]=[b_r]$. For the M\"{o}bius projection, the Paschke path of unitaries is $w_\theta=R(\theta)$, so $w_1$ is minus the identity. Then 
\begin{equation*}
    \gamma^{0,\rm Paschke}_{\rm id}[P_{\rm Mob}]=[-LP_{\rm Mob}(0)+1_2-P_{\rm Mob}(0)]=\left[\begin{pmatrix} -z & 0 \\ 0 & 1_{C(\TT,\varsigma)}\end{pmatrix}\right]\\
    =[-z]=[-1_{C(\TT,\varsigma)}]+[z],
\end{equation*}
so $\gamma^{0,\rm Paschke}_{\rm id}([P_{\rm Mob}]-[1_{C(\TT;\RR)}])=[-1_{C(\TT,\varsigma)}]$, i.e.\ the torsion generators are correctly mapped to each other.

The normalization axiom in the real $C^*$-algebra case needs to be replaced by
\begin{axioms}[Normalization]\label{normalizationreal}
If $(\mathcal{A},\alpha)=(\RR,\mathrm{id})$, then $\gamma^0_{\mathrm{id}}:KO_0(C(\TT;\RR))\cong KO^0(\TT)\rightarrow KO_1(C_\RR^*(\ZZ))\cong KO_1(C(\TT;\varsigma))$ takes $[1_{C(\TT;\RR)}]$ to the (Bott) generator $[b_r]$ and $([P_{\rm Mob}]-[1_{C(\TT;\RR)}])$ to $[-1_{C(\TT,\varsigma)}]$.
\end{axioms}

\begin{theorem}\label{realuniqueness}
There is a unique natural transformation $\gamma$ of the functors $(\mathcal{A},\alpha)\mapsto KO_\bullet(\Ind(\mathcal{A},\alpha))$ and $(\mathcal{A},\alpha)\mapsto KO_{\bullet+1}(\mathcal{A}\rtimes_\alpha\ZZ)$, $\bullet\in\ZZ_8$ , from the category of real $\ZZ$-algebras to abelian groups, which satisfy Axiom \ref{normalizationreal} and Axioms \ref{naturality}-\ref{suspension} with $K$ replaced by $KO$. Thus the real version of Corollary \ref{PaschkeequalsConnesThom} holds:  $\gamma^{\bullet,\,\rm{Paschke}}_\alpha=M_\alpha^{\bullet+1,\,\rm Green}\circ\phi^\bullet_{\tau^\alpha}$ up to a sign convention.
\end{theorem}

\section{Appendix: Multipliers for lattice groups}\label{sec:discretemultiplier}
Let $N=\ZZ^d$ and for a multiplier $\sigma$, let $\wt{\sigma}(n,n^\prime) = \sigma(n,n^\prime)/ \sigma(n^\prime,n)$ be the antisymmetric bicharacter labelling its class in $H^2(N,\TT)$. We reconstruct a canonical representative multiplier in this class starting from $\wt{\sigma}$. Choose generators $\{e_j\}_{j\in J}$, where $J$ is a totally ordered set, and, expanding $n = \sum_j n_je_j \in N$, and $n^\prime$ similarly, we define the multiplier 
$$
\sigma_J(n,n^\prime) = \prod_{j<k}\wt{\sigma}(e_j,e_k)^{n_jn_k^\prime},
$$
which is also a bicharacter (but not antisymmetric). Then we see that
\begin{eqnarray*}
\wt{\sigma}_J(n,n^\prime) &=& \frac{\sigma_J(n,n^\prime)}{\sigma_J(n^\prime,n)}\\ 
&=& \prod_{j<k}\wt{\sigma}(e_j,e_k)^{n_jn_k^\prime}\prod_{k<j}\wt{\sigma}(e_k,e_j)^{-n_jn_k^\prime}\\ 
&=& \prod_{j\neq k}\wt{\sigma}(e_j,e_k)^{n_jn_k^\prime}\\ 
&=& \wt{\sigma}(n,n^\prime).
\end{eqnarray*}
Thus we see that $\wt{\sigma}_J = \wt{\sigma}$, and consequently $\sigma_J$ is cohomologous to $\sigma$, and so also to any $\sigma_{J^\prime}$ for any other totally ordered set $J^\prime$.
We note that when $N = \integer^{2}$ and $J = \{1,2\}$ indexes the usual generators with the natural ordering, the antisymmetric bicharacter $\wt{\sigma}(n,n^\prime) = \exp(i(n_1n_2^\prime - n_2n_1^\prime))$ produces $\sigma_J(n,n^\prime) = \exp(i(n_1n_2^\prime))$.

%
%


\end{document}